\documentclass[11pt]{article}
\usepackage{amsmath,amssymb,amsfonts,amsthm,epsfig}
\usepackage{cancel}

\usepackage[usenames,dvipsnames]{xcolor}
\usepackage{bm,xspace}
\usepackage{tcolorbox}
\usepackage{cancel}
\usepackage{fullpage}
\usepackage{pras}
\usepackage{framed}
\usepackage{verbatim}
\usepackage{enumitem}
\usepackage{array}
\usepackage{multirow}
\usepackage{afterpage}
\usepackage{mathrsfs}
\usepackage{pifont} 
\usepackage{chngpage}
\usepackage[normalem]{ulem}
\usepackage{boxedminipage}
\usepackage{caption}
\usepackage{tikz}
\usepackage{mdframed}
\usepackage{setspace}
\usepackage{multicol}
\usepackage{titling}
\newcommand*\dif{\mathop{}\!\mathrm{d}}

\def\anonymizeymize{0}

\ifnum\anonymizeymize=1
\newcommand{\anonymize}[1]{}

\fi

\ifnum\anonymizeymize=0
\newcommand{\anonymize}[1]{#1}

\fi

\def\comments{1}

\ifnum\comments=1
\newcommand{\pras}[1]{\textcolor{BrickRed}{\sf{[#1 --PR]}}}
\newcommand{\kangning}[1]{\textcolor{orange}{\sf{[#1 --KW]}}}
\newcommand{\moses}[1]{\textcolor{green}{\sf{[#1 --MC]}}}

\fi

\ifnum\comments=0
\newcommand{\pras}[1]{}
\newcommand{\kangning}[1]{}
\newcommand{\moses}[1]{}
\fi

\def\colorful{0}

\ifnum\colorful=1

\fi
\ifnum\colorful=0

\fi

\DeclareMathOperator{\plu}{plu}
\DeclareMathOperator{\SC}{SC}
\DeclareMathOperator{\topC}{top}

\newcommand{\cg}{\succ}

\newcommand{\cgeq}{\succeq}

\renewcommand{\hat}[1]{\widehat{#1}}

\renewcommand{\sf}[1]{\textsf{#1}}

\makeatletter
\newtheorem*{rep@theorem}{\rep@title}
\newcommand{\newreptheorem}[2]{
\newenvironment{rep#1}[1]{
 \def\rep@title{#2 \ref{##1}}
 \begin{rep@theorem}\itshape}
 {\end{rep@theorem}}}
\makeatother
\theoremstyle{plain}

\newreptheorem{theorem}{Theorem}

\makeatletter
\newtheorem*{rep@claim}{\rep@title}
\newcommand{\newrepclaim}[2]{
\newenvironment{rep#1}[1]{
 \def\rep@title{#2 \ref{##1}}
 \begin{rep@claim}\itshape}
 {\end{rep@claim}}}
\makeatother
\theoremstyle{plain}

\newrepclaim{claim}{Claim}

\newtheorem{Alg}{Algorithm}

\begin{document}

\definecolor{myblue}{rgb}{0.15, 0.1, 0.75}
\definecolor{mypink}{rgb}{0.15, 0.55, 0.1}

\hypersetup{
    linkcolor = mypink,
    citecolor = myblue,
}

\title{
Breaking the Metric Voting Distortion Barrier \vspace{8pt}
}

\author{
\anonymize{
Moses Charikar\\\hspace{0pt}{{\sl Stanford}}
\and Prasanna Ramakrishnan \\\hspace{0pt}{{\sl Stanford}}
\and Kangning Wang\\\hspace{0pt}{{\sl Rutgers}}
\and Hongxun Wu\\\hspace{0pt}{{\sl UC Berkeley}}
}
}
\anonymize{
{\let\thefootnote\relax\footnotetext{Emails: \texttt{moses@cs.stanford.edu}, ~\texttt{pras1712@stanford.edu}, ~\texttt{kn.w@rutgers.edu}, ~\texttt{wuhx@berkeley.edu}.}}
}

\date{}

\pagenumbering{gobble}
\thispagestyle{empty}
\maketitle

\anonymize{
\vspace{0pt}
\hrule
\vspace{0pt}
}

\begin{abstract} 
We consider the following well-studied problem of metric distortion in social choice. Suppose we have an election with $n$ voters and $m$ candidates located in a shared metric space. We would like to design a voting rule that chooses a candidate whose average distance to the voters is small. However, instead of having direct access to the distances in the metric space, the voting rule obtains, from each voter, a ranked list of the candidates in order of distance. Can we design a rule that regardless of the election instance and underlying metric space, chooses a candidate whose cost differs from the true optimum by only a small factor (known as the \emph{distortion})?

A long line of work culminated in finding optimal deterministic voting rules with metric distortion $3$. However, for randomized voting rules, there is still a significant gap in our understanding: Even though the best lower bound is substantially lower at $2.112$, the best upper bound is still $3$, which is attained even by simple rules such as Random Dictatorship. Finding a randomized rule that guarantees distortion $3 - \eps$ for some constant $\eps$ has been a major challenge in computational social choice, as prevalent approaches to designing voting rules are known to be insufficient. In particular, such a voting rule must use information beyond aggregate comparisons between pairs of candidates, and cannot only assign positive probability to candidates that are voters' top choices.

In this work, we give a rule that guarantees distortion less than $2.753$. To do so we study a handful of voting rules that are new to the problem. One is \emph{Maximal Lotteries}, a rule based on the Nash equilibrium of a natural zero-sum game which dates back to the 60's. The others are novel rules that can be thought of as hybrids of Random Dictatorship and the Copeland rule. Though none of these rules can beat distortion $3$ alone, a careful randomization between Maximal Lotteries and any of the novel rules can.

\end{abstract}

\anonymize{
\vspace{10pt}
\hrule
\vspace{0pt}

\renewcommand\contentsname{\vspace{-18pt}} %

\begingroup
\let\clearpage\relax
{
\scriptsize
\begin{multicols}{2}
  \tableofcontents
\end{multicols}
}
\endgroup
}

\clearpage
\pagenumbering{arabic}


\section{Introduction}

\newcommand{\citetstar}[1]{\cite{#1}}
\newcommand{\citepstar}[1]{\cite{#1}}
\newcommand{\citetnostar}[1]{\cite{#1}}
\newcommand{\citepnostar}[1]{\cite{#1}}

Elections are a fundamental primitive in societal decision making. Through votes, people express their preferences and make collective decisions for social good. A common example of voting is single-winner elections, {in which} voters select one winner from a pool of candidates. These candidates can be persons that represent the voters, or broader social options such as potential locations to build a public facility. Voting is also applicable to {everyday} situations, such as choosing one from many lunch options for a group of colleagues, or picking a game to play among a group of friends. A \emph{voting rule} (or \emph{social choice rule}) maps the voters' preferences to a winning candidate.

A standard approach to evaluate the outcomes is to adopt the notion of utilitarian social efficiency. We assume that each voter has a cardinal utility function that maps each possible outcome to a real number quantitatively representing their preference for that outcome. From this utilitarian point of view, the optimal voting rule selects the outcome that optimizes the sum of the utilities.

In what has been classically studied and practically implemented, voting rules are often based on ordinal rankings{---}these rules make decisions based only on each voter's preference ordering, not the cardinal utilities, on the candidates. There are several {reasons why ordinal voting rules are appealing and widespread}. First, the restriction to ordinal rules simplifies the processes and the infrastructures needed for voting. Moreover, even though voters are assumed to have cardinal utilities, they may not be able to articulate them accurately, especially when these numbers represent differences in political stances in an abstract way. Finally, in a sense, this restriction to a common ordinal format gives each voter equal voting power.

Ordinal voting rules cannot always perfectly optimize social efficiency. Aiming to quantify the drawback of this format restriction{---}or this information loss from the perspective of social optimization{---}researchers have proposed the powerful notion of \emph{distortion} \citepnostar{DBLP:conf/cia/ProcacciaR06,DBLP:journals/ai/BoutilierCHLPS15,boutilier2016incomplete,DBLP:journals/ai/AnshelevichBEPS18}: It represents the worst-case ratio between the optimal efficiency and the efficiency of a particular ordinal voting rule (or in some contexts, the distortion-optimal ordinal voting rule). The worst-case distortion is generally not bounded by any constant, even after imposing normalization constraints on the cardinal utilities. This naturally calls for structural restrictions for the model.

In {their seminal work}, \bname{Anshelevich, Bhardwaj, Elkind, Postl, and Skowron} \citetstar{DBLP:conf/aaai/AnshelevichBP15,DBLP:journals/ai/AnshelevichBEPS18} proposed the influential framework of \emph{metric distortion}{---}in particular, they imposed the assumption that the voters and the candidates lie in a shared (unknown) metric space, and a voter's cardinal cost for a candidate is the distance between them in the metric space. {The assumption of metric preferences is a clean mathematical framework that stems from the large body of work on spatial models of voting \cite{enelow1984spatial,enelow1990advances,merrill1999unified,jessee2012ideology,armstrong2020analyzing}, in which voters and candidates have positions in a political spectrum, and each voter prefers nearer candidates in the spectrum. The metric assumption also captures situations when candidates are public facilities and voters' costs are their travel costs to the selected facility.}

More formally, the metric distortion of a social choice rule is defined as the supremum of the ratio between the social cost (i.e., sum of costs of voters) of this rule and the optimal social cost, over all possible metric spaces and all induced ordinal preference profiles. {While spatial models and in particular the metric assumption have their limitations and cannot perfectly capture most real-world scenarios, it has been observed that spatial models can be more accurate than other preference models \cite{tideman2011modeling}.
Aside from the credibility of the metric assumption in practice, it is still a strong theoretical guarantee that a low-distortion voting rule must perform well with respect to \emph{any} metric space. A priori, it might seem that such a strong guarantee could be impossible, but in fact, the metric structure reduces the distortion of many social choice rules to \emph{constants}, leading to the intriguing search for optimal social choice rules through the lens of metric distortion.}

\begin{figure}[!ht]
\centering
\includegraphics[scale=1]{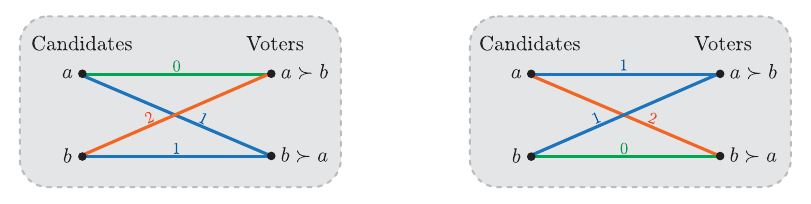}
\caption{In an election with two candidates and two disagreeing voters, both of the above metric spaces are possible. No deterministic rule can have distortion less than 3, and no randomized rule can have distortion less than 2.%
}
\label{fig:2v2c}
\end{figure}

\paragraph{The journey to distortion 3 for deterministic rules.} One fruitful line of work on metric distortion, including the original one of \citetstar{DBLP:conf/aaai/AnshelevichBP15,DBLP:journals/ai/AnshelevichBEPS18}, focuses on deterministic social choice rules. An immediate lower bound for deterministic rules is $3$ (see \Cref{fig:2v2c}): There are two candidates $\{a, b\}$ and two voters {such that} one voter prefers $a$ to $b$ and the other prefers $b$ to $a$. Choosing either candidate gives distortion of $3$. \citetstar{DBLP:conf/aaai/AnshelevichBP15,DBLP:journals/ai/AnshelevichBEPS18} also showed that any voting rule selecting from the \emph{uncovered set} (which we will introduce later along with its generalizations), such as the Copeland rule, guarantees an upper bound of $5$. This gap between $3$ and $5$ was a tantalizing one and resisted researchers' efforts (e.g.{,} \citepstar{DBLP:conf/sigecom/GoelKM17}) for a few years. \bname{Munagala and Wang} \citetnostar{DBLP:conf/ec/MunagalaW19} were the first to reduce the gap: They proposed a novel weighted variant of the uncovered set to improve the upper bound to $2 + \sqrt{5} \approx 4.236$, and also showed that selecting from a novel \emph{matching uncovered set} guarantees a distortion upper bound of $3$. However, they did not manage to show that the set is always non-empty, and left it as a conjecture. This conjecture had been studied by many researchers (e.g.{,} \bname{Kempe} \citepnostar{DBLP:conf/aaai/Kempe20a} who identified alternative and additional formulations) since then, until \bname{Gkatzelis, Halpern, and Shah} \citetstar{DBLP:conf/focs/GkatzelisHS20}, in their breakthrough result, proved it true. \citetstar{DBLP:conf/focs/GkatzelisHS20} identified the crux of the conjecture and proved the existence of a candidate that satisfies their new simplified conditions, hence showing a distortion upper bound of $3$ for both their novel rule of Plurality Matching and the one of \citetnostar{DBLP:conf/ec/MunagalaW19}. This {result} closes the gap for deterministic social choice rules. In the impressive work of \bname{Kizilkaya and Kempe} \citetnostar{DBLP:conf/ijcai/KizilkayaK22}, they proposed the novel, elegant, and practical voting rule Plurality Veto, which is also by itself a one-paragraph constructive proof of the conjecture of \citetnostar{DBLP:conf/ec/MunagalaW19}. \bname{Kizilkaya and Kempe} \citetnostar{DBLP:conf/sigecom/Kizilkaya023} further identified other related practical voting rules with distortion $3$ {that have connections to} the \emph{proportional veto core} \citepnostar{moulin1981proportional}, a classical notion in social choice.

\paragraph{The need for randomization and the barrier of 3.} The above canonical lower bound of $3$ only involves two candidates and two disagreeing voters, and one factor that induces this ``large'' distortion is the limit to \emph{deterministic} social choice rules. This limit is traditionally motivated by people's aversion to randomization for important social issues{, for which ``bad luck'' could be calamitous. 
Though this is a valid reason for sticking to deterministic rules, we point out that randomization is acceptable in some situations. In a tied election such as \Cref{fig:2v2c}, it might be more natural to randomize over the candidates, than to break the symmetry arbitrarily to pick one candidate deterministically. In fact, in real-world elections, it is common to use randomization to break ties. Beyond tie-breaking, randomization has also been used for other purposes in both practical and theoretical social choice settings. Some prominent examples include choosing citizens' assemblies via sortition (dating back to ancient Athenian democracy \cite{headlam1891election}) and fair allocation of goods using a lottery \cite{abdulkadirouglu1998random}. Another perspective is that from our utilitarian point of view, randomized voting rules can be considered equivalent to fractional voting rules, which are acceptable in some situations.
One example is allocating funding to the candidate projects in participatory budgeting. Another example is when rotation of authority is involved, such as the consulship in ancient Rome where two consuls were elected annually and switched their roles and responsibilities every month \cite{abbott1901history}. We refer the reader to \cite{brandt2017rolling} and the references therein for a more in-depth discussion of randomization in social choice.}

Soon after the first work on metric distortion, \bname{Feldman, Fiat, and Golomb} \citetstar{DBLP:conf/sigecom/FeldmanFG16} and \bname{Anshelevich and Postl} \citetnostar{DBLP:journals/jair/AnshelevichP17} independently studied metric distortion without the restriction to deterministic rules. They both showed an upper bound of $3$ for the simple rule of Random Dictatorship (which outputs the favorite candidate of a uniformly random voter), and gave a lower bound of $2$ using {the} same two-voter-two-candidate example, leaving open the possibility that much better metric distortion could be achieved by randomized voting rules. This gap between $2$ and $3$ {has been} a very (if not the most) intriguing question in the field of distortion.

The lack of progress on this question motivated researchers to look at fine-grained distortion analysis within the instance classes of a fixed number of voters or candidates. For example, \citetnostar{DBLP:journals/jair/AnshelevichP17} showed that Random Dictatorship has distortion $3 - 2/n$ within the instance class of $n$ voters, for any $n$. Several other works provided upper bounds for the instance class of $m$ candidates, for any $m$: \citetstar{DBLP:conf/aaai/FainGMP19} proposed a rule called Random Oligarchy that can achieve an upper bound slightly worse than $3 - 2/m$;\footnote{Their upper bound is $3 - 2 \cdot \min_{p \in [0, 1]} \left(p^2 (2 - p) + (1 - p)^3 / (m - 1)\right)$, which is $3 - 2/m + O(1/m^2)$ as $m \to \infty$.} \citetnostar{DBLP:conf/aaai/Kempe20b} first showed an upper bound of $3 - 2/m$ by mixing Random Dictatorship with a rule named Proportional to Squares; \citetstar{DBLP:conf/focs/GkatzelisHS20} proposed Smart Dictatorship, a variant of Random Dictatorship, which gives the same guarantee of $3 - 2/m$. These improvements over $3$ vanish when we consider the supremum over all instances.

The first constant improvement for this gap is on the lower bound. \bname{Charikar and Ramakrishnan} \citetnostar{DBLP:conf/soda/CharikarR22} improved the lower bound to $2.1126$, while \bname{Pulyassary and Swamy} \citetnostar{pulyassary2021randomized} independently showed a lower bound of $2.0631$. For the upper bound, there have been many different rules with distortion $3$, but also many classes of rules which are known to be \emph{unable} to beat $3$: deterministic rules \citepstar{DBLP:journals/ai/AnshelevichBEPS18}, rules that only {consider} the top choices of the voters \citepstar{DBLP:conf/aaai/GrossAX17}, and rules that only look at pairwise comparisons of candidates (i.e.{,} weighted tournament rules) \citepstar{DBLP:conf/sigecom/GoelKM17}. These results rule out a large swath of voting rules that have been studied in the metric distortion literature, which raises the following pressing question.

\begin{question}
Is there a voting rule with metric distortion better than $3$? What might such a voting rule look like?
\end{question}

In this work, we answer this question %
by showing that a randomization over simple rules can achieve distortion less than $2.753$.

\subsection{Our Techniques and Voting Rules}

\paragraph{The biased metric framework.} Our work uses a new analysis framework that refines the linear programming approach introduced by \citetnostar{DBLP:conf/soda/CharikarR22}. Each metric can be viewed as a linear constraint on a potential voting rule, and {the insight of} \citetnostar{DBLP:conf/soda/CharikarR22}  was to show that a relatively simple class of metrics, called \emph{biased metrics}, characterizes the {strictest} constraints.
However, they were only able to analyze biased metrics with some relaxations, and could only prove upper bounds for elections with three candidates. Our approach, on the other hand, allows us to precisely characterize the constraints imposed by biased metrics. The resulting framework gives us more analysis power while retaining most of the simplicity, and gives us the intuition that {enables us to} break of the barrier of $3$.

In the main body of the paper, we consider three randomized rules and analyze them and their mixtures using this framework. In \Cref{sec:revisit-known}, we also use this framework to revisit a variety of results proved in the metric distortion literature \cite{DBLP:conf/aaai/AnshelevichBP15, DBLP:conf/sigecom/FeldmanFG16, DBLP:journals/jair/AnshelevichP17, DBLP:conf/ec/MunagalaW19,
DBLP:conf/focs/GkatzelisHS20,
DBLP:conf/ijcai/KizilkayaK22} 
and show that they have short, simple proofs once the biased metric framework has been established. {These proofs suggest} that the framework may be a helpful primitive for future work in the area.

\paragraph{A note on weighted tournament rules.} 

\emph{Weighted tournament rules} (or \emph{C2 rules} \citepnostar{fishburn1977condorcet}) are a special class of voting rules that only consider pairwise comparisons (for each pair of candidates $i, j$, the proportion of voters that prefer $i$ over $j$). Weighted tournament rules are desirable in many settings since they are often simple, interpretable, and efficiently implementable by sampling voters. 

The Maximal Lotteries rule discussed in \Cref{sec:ml} is a weighted tournament rule, and is in fact optimal among such rules. Our other rules, including Random Consensus Builder (\Cref{ssec:rcb}), Random Dictatorship on the (Weighted) Uncovered Set (\Cref{ssec:radius}), and Random Dictatorship on the {Quasi-Kernel} (\Cref{sec:indep-set}), are ``almost'' weighted tournament rules: They only consider pairwise comparisons of the candidates and the ordering of a uniformly random voter. Note that this modification still allows the rules to be efficiently implemented with sampling.

\paragraph{Maximal Lotteries.} The first rule we study is \emph{Maximal Lotteries}. According to \citetstar{brandt2017rolling}, it (along with its variants) was first considered by \bname{Kreweras} \citetstar{kreweras1965aggregation}, independently rediscovered and studied in detail by \bname{Fishburn} \citetstar{fishburn1984probabilistic}, and later also independently rediscovered by \citetstar{laffond1993bipartisan, fisher1995tournament, felsenthal1992after, rivest2010optimal}. This rule has not been studied in the context of distortion.

Maximal Lotteries formulates the following zero-sum game: Two players 1 and 2 each {choose} a distribution over the candidates, and then {nature} independently draws a candidate $c_1$ from Player 1's distribution, a candidate $c_2$ from Player 2's distribution, and a uniformly random voter $v$. Player $1$ wins if $v$ prefers $c_1$ to $c_2$ and vice versa, breaking ties uniformly in case $c_1 = c_2$. Each player aims to maximize their winning probability. Maximal Lotteries outputs a Nash equilibrium of this zero-sum game, which can be computed in polynomial time.

We show that Maximal Lotteries has distortion $3$. This {result} on its own resolves an interesting question on the optimal distortion of weighted tournament rules. \bname{The work of} \cite{DBLP:conf/sigecom/GoelKM17} showed that no such rule can have distortion better than $3$, and the best previously known upper bound was $2 + \sqrt{5}$ \cite{DBLP:conf/ec/MunagalaW19}. Additionally, our framework admits a finer characterization on the worst-case instances: {Intuitively, we can show that in any metric space where Maximal Lotteries has distortion close to 3, nearly half of the voters must be tightly clustered around the optimal candidate. But then, if much more than half the voters prefer $c_1$ over $c_2$, then $c_1$ cannot be much farther from the optimal candidate than $c_2$, since some voter in the tight cluster prefers $c_1$ over $c_2$. Our strategy is to design rules that can leverage this structure, and perform particularly well on the instances for which Maximal Lotteries performs poorly.}

\paragraph{Random Consensus Builder.}

Motivated by the discussion above, we conceptually build a directed graph {in which} the vertices are the candidates. We draw an edge from $c_1$ to $c_2$ if $c_1$ beats $c_2$ with a large margin in their pairwise comparison. We are inspired by the following graph theory fact: In any directed graph, there exists an independent set, so that any vertex in the graph can be reached from a vertex in the independent set in at most two steps {\cite{chvatal1974every}}. {(Such an independent set is called a quasi-kernel.)} When Maximal Lotteries is ``bad'', a candidate in this independent set must be close to the true optimal candidate. Additionally, candidates in this independent set must be relatively even in their pairwise comparisons due to our construction of the graph. These additional structures make Random Dictatorship on the independent set perform much better than distortion $3$ in these cases.

Our Random Consensus Builder rule utilizes this intuition but only implicitly picks this independent set.\footnote{For the interested reader, another voting rule, Random Dictatorship on the {Quasi-Kernel}, which more directly uses this idea, is discussed in \Cref{sec:indep-set}.} Random Consensus Builder picks a uniformly random voter and looks at the candidates from her least preferred one to her most preferred one. When a candidate $c$ is encountered, all candidates that $c$ can pairwise beat with a large margin are removed. The last candidate encountered during the process is the winner. Conceptually, Random Consensus Builder naturally balances the opinion of a random voter with the general consensus.

Using our framework, we show that a randomization between Maximal Lotteries and Random Consensus Builder with proper parameters has distortion at most $2 \sqrt{2} \approx 2.82843$.

\paragraph{RaDiUS: \underline{Ra}ndom \underline{Di}ctatorship on the (Weighted) \underline{U}ncovered \underline{S}et.}

Our analysis of Random Consensus Builder uses properties that are reminiscent of the weighted uncovered set, proposed by \citetnostar{DBLP:conf/ec/MunagalaW19} who showed that an arbitrary selection from this set (with a proper parameter) gives distortion $2 + \sqrt{5} \approx 4.23607$. This {connection} motivates us to propose RaDiUS (\underline{Ra}ndom \underline{Di}ctatorship on the (Weighted) \underline{U}ncovered \underline{S}et) that outputs a uniformly random voter's favorite candidate within the weighted uncovered set.

It turns out RaDiUS can give better guarantees than Random Consensus Builder. Using our framework, we show that a randomization between Maximal Lotteries and RaDiUS with proper parameters has distortion at most $2.75271$.

\subsection{Further Related Work}

There has been a large body of work on distortion in social choice. We refer the reader to the survey of \citetstar{DBLP:conf/ijcai/AnshelevichF0V21} for a more detailed overview of the field; below we briefly discuss some of them.

The first works on distortion did not impose the metric-space condition, {and instead assumed that the} utilities are non-negative, and defined distortion of a rule as the worst-case ratio between the optimal sum of utilities and the sum of utilities attained by the rule \citepnostar{DBLP:conf/cia/ProcacciaR06}. Many works made the unit-sum utility assumption, {in which every voter's total utility over all candidates} equals $1$, to avoid uninteresting worst cases. Under this assumption, {a variety of upper and lower bounds have been proved, parameterized by the number of candidates in an instance, $m$.} \bname{The work of} \citetstar{DBLP:journals/ai/CaragiannisP11} showed that the Plurality rule has distortion $O(m^2)$. A matching $\Omega(m^2)$ lower bound for deterministic rules was later given by \citetstar{DBLP:journals/jair/CaragiannisNPS17}. Under the same assumption, \bname{the work of} \citetstar{DBLP:journals/ai/BoutilierCHLPS15} proposed a randomized rule with distortion $O(\sqrt{m} \log^* m)$ and gave a lower bound of $\Omega(\sqrt{m})$ for any rule. This gap was closed by \citetstar{DBLP:journals/teco/EbadianKPS24} who proposed a Stable Lottery (and Stable Committee) rule, which was inspired by fair committee selection literature, with distortion $O(\sqrt{m})$.

\bname{The work of} \citetstar{DBLP:conf/sigecom/GkatzelisL023} aimed to provide best-of-both-worlds guarantees for both the metric setting and the non-metric setting. They proposed novel deterministic and randomized social choice rules which guarantee constant metric distortion and almost optimal (for deterministic and randomized rules correspondingly) non-metric distortion.

Researchers have also looked beyond single-winner elections {which can only select} one winner from the candidates. \bname{The work of} \citetstar{DBLP:journals/ai/CaragiannisSV22} considered a model of multi-winner elections in the metric distortion setting, {for which} they give complete characterizations for the optimal metric distortion. Graph problems such as selecting a perfect bipartite matching (e.g.{,} \citepnostar{DBLP:conf/aaai/AnshelevichS16,DBLP:journals/teco/AnshelevichZ21,DBLP:journals/jair/AmanatidisBFV22,DBLP:conf/sigecom/AnariCR23,caragiannis2024truthful}) in both metric and non-metric distortion settings have also received great attention.

Most works in this field aim to optimize the utilitarian aggregation of preferences, i.e., the sum or average of the utilities. Other works consider ``fair'' ways to aggregate preferences: \bname{The work of} \citetstar{DBLP:conf/sigecom/GoelKM17} proposed the \emph{fairness ratio} in the metric setting, which is inspired by the mathematical idea of majorization and replaces the utilitarian aggregation by the worst-case symmetric monotonic norm. \citetstar{DBLP:conf/sigecom/GoelKM17} showed a lower bound of $3$ and an upper bound of $5$ {on the optimal fairness ratio, and} \citetstar{DBLP:conf/focs/GkatzelisHS20} closed this gap by showing their Plurality Matching rule has a fairness ratio of $3$. \bname{The work of} \citetstar{DBLP:journals/teco/EbadianKPS24} studied proportional fairness, Nash welfare, and the core in the non-metric setting, and gave {tight} distortion bounds of $O(\log m)$ for all these objectives.

The distortion framework serves as a valuable tool to quantify the efficiency of voting rules, and therefore has been adopted in the study of various aspects of voting, such as the tradeoff between the amount of communication and the efficiency performance of voting rules \citepstar{DBLP:conf/aaai/GrossAX17,DBLP:conf/aaai/FainGMP19,DBLP:conf/nips/MandalPSW19,DBLP:conf/aaai/Kempe20b,DBLP:conf/sigecom/MandalSW20}. The framework is also used to quantify the effect of certain social structures: \bname{The work of} \citetstar{DBLP:conf/sigecom/ChengDK17,DBLP:conf/aaai/ChengDK18} studied the representativeness of candidates on the population of voters. In particular, they showed that when the candidates are drawn independently from the voter population, the metric distortion of social choice rules becomes much better. \bname{The work of} \citetstar{DBLP:conf/sigecom/FlaniganPW23} studied the effect of public spirit in the non-metric distortion framework. They showed that if every voter altruistically ranks the candidate according to a mixture of her own preference and the average preference of the voters, then the distortion of many social choice rules will drastically improve.

\section{Preliminaries and Notation}\label{sec:prelims}

\paragraph{Elections.} An \emph{election instance} is defined by a tuple $\mathcal{E} = (V, C, \cg_V)$, where $V$ is a set of $n$ voters, $C$ is a set of $m$ candidates, and $\cg_V$ {maps each voter $v$ to a linear order over the candidates $\cg_v$,} {such that} $i \cg_v j$ if voter $v$ prefers candidate $i$ over candidate $j$. Throughout the paper, we will use $i, j, k, a, b, c$ to refer to candidates and $u,v$ to refer to voters. We will have $i^*$ denote the true best candidate when the metric space is fixed.

For a condition $\mathcal{P}$, we let $S_{\mathcal{P}}$ denote the subset of voters whose preferences satisfy $\mathcal{P}$. We also let $s_{\mathcal{P}} = |S_{\mathcal{P}}|/n$ be the proportion of these voters overall, or equivalently, the probability that a uniformly random voter's preference list satisfies property $\mathcal{P}$. For example, $S_{i \cg j}$ is the set of voters that prefer $i$ over $j$, $S_{i, j \cg k}$ is the set of voters that prefer $i$ and $j$ over $k$, and $S_{I \cg j}$ is the set of voters that prefer all the candidates in $I$ over $j$. Note that if $j \in I$ then $S_{I \cg j} = \varnothing$ and $s_{I\cg j} = 0$. We also let $\plu(i) = s_{i\cg C\setminus \{i\}}$ be the proportion of voters whose first choice is $i$. 

In \Cref{sec:ml}, we will also allow the property $\mathcal{P}$ to be randomized, in which case {$s_{\mathcal{P}}$ is treated as the expected fraction of voters in the set-valued random variable $S_{\mathcal{P}}$.} 
For example, if $D$ is a distribution over candidates, then $s_{D \cg j}$ denotes the probability that a uniformly random voter prefers a candidate $i \sim D$ over $j$. That is, %
$$s_{D \cg j} =\Prx_{i \sim D, v\sim V} [i\cg_v j] = \Ev_{i \sim D, v\sim V}[\textbf{1}[i\cg_v j]].$$
Note that in the case that $i = j$, it will be natural to treat $\textbf{1}[i\cg_v j]$ as $\frac12$. To this end, in \Cref{sec:ml} we will let $s_{i\cg i}$ be $\frac12$ instead of 0{---}this makes it so that the \emph{Condorcet Game} is well defined, and it makes the proof of \Cref{thm:ML-ub} read much more smoothly.

We use this notation extensively and flexibly in the paper, and we may reiterate what certain instances mean in natural language to be clear.

\paragraph{Metric spaces.} A metric space is a pair $(\mcal{M}, d)$ of a set $\mcal{M}$ and a distance metric $d: \mcal{M} \times \mcal{M} \to \mathbb{R}_{\geq 0}$ with the following three properties:
\begin{enumerate}[label=(\arabic*)]
	\item Positive definiteness: $d(x, y) \geq 0$ with equality if and only if $x = y$, 

	\item Symmetry: $d(x, y) = d(y, x)$,

	\item Triangle inequality: $d(x, y) \leq d(x, z) + d(z, y)$.

\end{enumerate}

We will extend the notation of $d$ to operate directly on the voters and candidates rather than the points they occupy in the metric space. Note that for simplicity we allow voters and candidates to be co-located in the space, so their distance may be zero. (That is to say, technically, we consider \emph{pseudometric spaces}. This simplification does not change the distortion of any voting rule.)

When defining biased metrics in \Cref{sec:biased} and proving lower bounds in \Cref{sec:lbs}, we will specify metric spaces {defining only} the distances between candidates and voters {explicitly}, and {leaving} the other distances implicit. We make no use of the implicit distances, but to fully specify the metric space one can use the \emph{graph distance closure} of the explicitly defined distances. i.e., if the distance between two points is not explicitly defined, it should be taken to be the shortest path between those two points using the explicitly defined distances.

Given an election instance $\mathcal{E} = (V, C, \cg_V)$, we say that a distance metric $d$ is \emph{consistent} with $\mathcal{E}$ (denoted $d \rhd \mathcal{E}$) if for all $v \in V$, $i \cg_v j$ implies $d(i, v) \leq d(j, v)$.

\paragraph{Voting rules, social cost, and distortion.} For an election with underlying distance metric $d$, we {define} the social cost of a candidate $i$ to be their average distance to the voters. i.e., 
$$\SC(i, d) := \frac1n \sum_{v \in V} d(i, v).$$
(In the literature, the social cost of a candidate is usually the \emph{sum}, rather than the average, of distances. Since we are concerned {with} the ratio between costs, we can equivalently use the average-distance version, which we find easier to work with.) We will often just write $\SC(i)$ when the relevant distance metric has been fixed, or is clear from context.

A \emph{voting rule} (or \emph{social choice rule}) $f$ is a function that maps every election instance $\mcal{E}$ to a distribution over its candidates. Given this, the \emph{distortion} of $f$ is given by
$$\mathrm{distortion}(f) = \sup_{\mcal{E}} \sup_{d: d\rhd \mcal{E}} \frac{\displaystyle\Ev_{j \sim f(\mcal{E})}[\SC(j, d)]}{\displaystyle\min_{i\in C} \SC(i, d)}. $$
We will also often refer to the distortion of $f$ on a particular metric $d$, which is just the operand of the suprema above. When the election instance and voting rule are fixed, we will use $p_j$ to denote the probability that the rule chooses candidate $j$ on the instance.

{In some discussions, we reference the particular class of \emph{weighted tournament rules}. Given an election instance, one can construct a weighted tournament graph whose vertices are the candidates $C$, and the weight of the edge $(i, j)$ is $s_{i\cg j}$. A weighted tournament rule is a voting rule that only uses the weights in this graph $\langle s_{i \cg j} \rangle_{i, j \in C}$ rather than the full ordinal rankings of the voters $\cg_V$. More precisely, a voting rule $f$ is a weighted tournament rule if, for any two election instances $\mcal{E} = (V, C, \cg_V)$ and $\mcal{E}' = (V, C, \cg_V')$ for which $\langle s_{i \cg j} \rangle_{i, j \in C}$ is the same, $f(\mcal{E}) = f(\mcal{E}')$. Some prominent examples include the \emph{Copeland rule}, which deterministically chooses the candidate $i$ that maximizes $|\{j: s_{i\cg j} \geq \frac12\}|$, the \emph{Borda count}, which deterministically chooses the candidate $i$ that maximizes $\sum_{j\neq i} s_{i\cg j}$, and \emph{Maximal Lotteries}, which is formally defined in \Cref{sec:ml}.}

\section{The Biased Metrics}\label{sec:biased}

The key tool that we use to understand the metric distortion of the social choice rules in \Cref{sec:mechs} is a refinement of the linear programming framework introduced by \cite{DBLP:conf/soda/CharikarR22}. 

Suppose that we have an election instance. If we can design a rule such that for any metric consistent with the instance we have
\begin{equation}\label{eq:linear-con}
\sum_{j\in C} (\SC(j) - \SC(i^*))p_j \leq \lambda \cdot 2\SC(i^*),
\end{equation}
then the rule has distortion at most $1 + 2\lambda$ on this instance. {(Although the factor of 2 can be absorbed into $\lambda$, it is convenient to keep it to cancel with the factor of $1/2$ in \Cref{def:biased}.)} %
In this view, for a rule to have low distortion it has to satisfy a set of linear constraints, with one constraint imposed by each metric. As one might expect, some constraints may be redundant, so it is helpful to try and find a small set of metrics whose constraints imply those for all of the metrics. Then, one can show that a rule has low distortion just by showing that it has low distortion on the small set of metrics. 

\cite{DBLP:conf/soda/CharikarR22} defined the set of \emph{biased metrics}, and showed that they {satisfy} this property.

\begin{definition}\label{def:biased}
Let $(x_1, \ldots, x_m)$ be a vector of nonnegative real numbers such that $x_{i^*} = 0$ for some $i^*$. Given an election instance, the \emph{biased metric} for the vector $(x_1, \ldots, x_m)$ is defined as follows. For a voter $v$ and candidate $i$, let
\begin{align*}
d(i^*, v) &= \frac12 \max_{i, j: i \cgeq_v j} (x_i - x_j),\\
d(j, v) - d(i^*, v) &= \min_{k: j \cgeq_v k} x_k.
\end{align*}
\end{definition}

The rough idea of the biased metrics is the following. Suppose we were given some fixed metric such that the distance from candidate $j$ to the optimal candidate is $x_j$ (so $x_{i^*} = 0$). Then we could imagine throwing out all of the other distances, and remaking them so that the distances from $i^*$ to the voters {are} as small as possible, and the distances from the other candidates {are} as large as possible (compared to the distances from $i^*$). The former is to make the right side of \cref{eq:linear-con} smaller and the latter is to make the left side larger, which will tighten the constraint. It turns out that {implementing this approach while} respecting the triangle inequality and the preferences {leads to} the above definition. 

\cite{DBLP:conf/soda/CharikarR22} gave proofs that biased metrics are indeed valid distance metrics, and that they tighten the constraints in \cref{eq:linear-con}. For completeness, these proofs are included in \Cref{sec:biased-pfs}.

Now that it has been established that we only need to consider the constraints imposed by biased metrics, let us see how to express these constraints. Suppose that we have a fixed biased metric given by a vector $(x_1, \ldots, x_m)$. Let $I_t = \{k \in C: x_k \leq t\}$. Notice then that $d(j, v) - d(i^*, v) > t$ if and only if $v \in S_{I_t \cg j}$. To be clear, note that if $j \in I_t$ then $S_{I_t \cg j} = \varnothing$ and $s_{I_t \cg j} = 0$. It follows that
$$\SC(j) - \SC(i^*) = \Ev_{v\sim V}[d(j, v) - d(i^*, v)] = \int_0^\infty\Prx_{v\sim V}[d(j, v) - d(i^*, v) > t] \dif t = \int_0^\infty s_{I_t \cg j} \dif t$$
and so 
$$\sum_{j\in C} (\SC(j) - \SC(i^*))p_j = \int_0^\infty \sum_{j\notin I_t}  s_{I_t \cg j} p_j \dif t.$$

We can use a similar approach to express $2\SC(i^*)$. We have that $2d(i^*, v) \leq t$ if and only if $v \in S_{\forall i\cg j,  x_i - x_j \leq t}$. This is the set of voters $v$ such that whenever $i\cg_v j$, we have $x_i - x_j \leq t$. It follows that
$$2\SC(i^*) = \Ev_{v \sim V}[2d(i^*, v)] = \int_0^\infty(1 - \Prx_{v \sim V}[2d(i^*, v) \leq t]) \dif t = \int_0^\infty(1 - s_{\forall i\cg j,  x_i - x_j \leq t}) \dif t.$$
Therefore, the constraint imposed by the biased metric is
\begin{equation}\label{eq:precise-biased}
\int_0^\infty \sum_{j\notin I_t}  s_{I_t \cg j} p_j \dif t \leq  \lambda\int_0^\infty(1 - s_{\forall i\cg j,  x_i - x_j \leq t}) \dif t.  
\end{equation}
{To show that a voting rule achieves distortion at most $1 + 2\lambda$ on a given election instance, it suffices to show that it satisfies the above constraint for each biased metric.}

{For} a sense of how the right side of this expression behaves, {we can show} that $s_{\forall i\cg j,  x_i - x_j \leq t} \leq s_{i^* \cg I_t^c}$. {It suffices to show that $S_{\forall i\cg j,  x_i - x_j \leq t} \subseteq S_{i^* \cg I_t^c}$. Suppose that $v \in S_{\forall i\cg j,  x_i - x_j \leq t}$, meaning that whenever $i \cg_v j$ we have $x_i - x_j \leq t$. It follows that if $k \cg_v i^*$, then $x_k \leq t$ (since $x_{i^*} = 0$), and so $v$ can only prefer candidates in $I_t$ over $i^*$. This {conclusion implies} that $v \in S_{i^*\cg I_t^c}$, as claimed.}

In a lot of situations, using $s_{i^* \cg I_t^c}$ in place of the more complicated expression is sufficient. To this end, {we have} the following constraint {which} implies \cref{eq:precise-biased}.

\begin{equation}\label{eq:weaker-biased}
\int_0^\infty \sum_{j\notin I_t}  s_{I_t \cg j} p_j \dif t \leq  \lambda\int_0^\infty(1 - s_{i^* \cg I_t^c}) \dif t.  
\end{equation}

\cite{DBLP:conf/soda/CharikarR22} derived \cref{eq:weaker-biased} (though written in a more discrete form), and then noted that one could consider an even stricter collection of constraints, which we will use to analyze Maximal Lotteries in \Cref{sec:ml}.
\begin{equation}\label{eq:fancy}
\sum_{j\notin I}  s_{I \cg j} p_j \leq  \lambda (1 - s_{i^*\cg I^c}).
\end{equation}
In particular, to show \cref{eq:precise-biased} for all possible metrics, it suffices to show the above inequality for all sets {$I \notin \{\varnothing, C\}$}, and all $i^* \in I$. The convenience of this is that there are finitely many possible constraints rather than the infinitely many constraints {one would have to consider if one used} \cref{eq:precise-biased} or \cref{eq:weaker-biased} (since each choice of the vector $(x_1, \ldots, x_m)$ may give a different constraint). %
Moreover, the metric space has been completely abstracted away{---}these constraints only involve terms that come from the election instance alone. However, it should be noted that \cref{eq:precise-biased} loses no generality (a rule has distortion $1 + 2\lambda$ if and only if it satisfies \cref{eq:precise-biased} for all biased metrics), but \cref{eq:weaker-biased} and \cref{eq:fancy} may lose generality (the implication only goes one way).

To conclude this section, we introduce some notation that makes \cref{eq:precise-biased} easier to discuss. %

\begin{figure}[!h]
\centering
\includegraphics[scale=0.8]{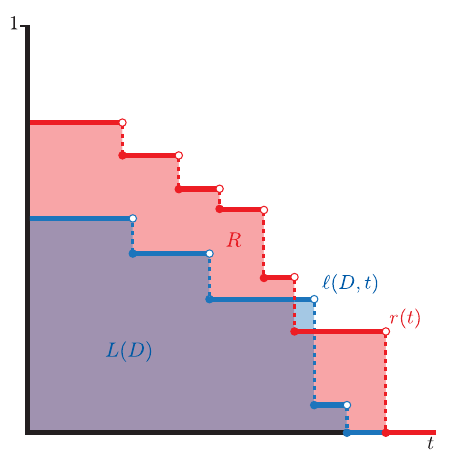}
\caption{An example of the functions $r(t)$, $\ell(D, t)$, and the areas $R$ and $L(D)$.}
\end{figure}

Once an election instance is fixed, we let
$$r(t) = 1 - s_{\forall i\cg j,  x_i - x_j \leq t} \qquad \text{and} \qquad R = \int_0^\infty r(t) \dif t.$$
Given a distribution $D$ over the candidates which chooses candidate $j$ with probability $p_j$, we let 
$$\ell(D, t) = \sum_{j\notin I_t} s_{I_t\cg j} p_j \qquad \text{and} \qquad L(D) = \int_0^\infty \ell(D, t) \dif t.$$

With this notation, we would like to design a rule which outputs a distribution $D$ such that for all biased metrics, $L(D)/R \leq \lambda$ for a small fixed $\lambda${. (T}o get distortion less than 3 we need $\lambda < 1${.)} Though the expression for $r(t)$ may seem complicated and difficult to work with, ultimately we only use it in two simple ways. As mentioned, in \Cref{sec:ml} we only use $r(t) \geq 1 - s_{i^* \cg I_t^c}$. In \cref{sec:mechs} it is only needed for \Cref{prop:small-r-consistent-metric}, which roughly says that if $r(t)$ is small, then the metric admits a nice structure that can be leveraged to get better distortion bounds.

\section{Maximal Lotteries}\label{sec:ml}

In this section, we will study the distortion of Maximal Lotteries \citepnostar{kreweras1965aggregation}. The voting rule is based on a zero-sum game (in our formulation, a constant-sum game) 
between two players Alice and Bob. The details of the game and the voting rule are below.

\begin{framed}
    \begin{center}
        \textbf{The Condorcet Game}
    \end{center}

    \begin{itemize}
        \item Simultaneously, Alice picks a distribution $D_A$ and Bob picks a distribution $D_B$ over the candidates. 

        \item {Sample} $a \sim D_A$ and $b \sim D_B$.

        \item Alice and Bob's payoffs are $s_{a\cg b}$ and $s_{b\cg a}$ respectively. If $a = b$ then each player gets $\frac12$.
         
    \end{itemize}

    \begin{center}
        \textbf{Maximal Lotteries (ML)}
    \end{center}

    \begin{itemize}

    	\item Choose a candidate from any Nash equilibrium distribution of the Condorcet game.

    \end{itemize}
\end{framed}

For the Condorcet game, note that under the notation we introduced in \Cref{sec:prelims}, once $D_A$ and $D_B$ are picked the payoffs for Alice and Bob are $s_{D_A \cg D_B}$ and $s_{D_B \cg D_A}$ respectively. We remind the reader that in this section we treat $s_{i \cg i}$ as $\frac12$ so that if both players choose the same candidate, their payoffs are equal.

Suppose that we fix an election instance. Let $D$ be the distribution output by ML, which chooses candidate $i$ with probability $p_i$. Let $P(I) = \sum_{i\in I} p_i$ be the probability that a candidate in $I$ is chosen. Let $D(I)$ be the distribution conditioned on the chosen candidate coming from $I$. i.e., a candidate $i \in I$ is chosen with probability $p_i/P(I)$. Note that $D(I)$ is not well-defined if $P(I) = 0$, so we will deal with cases {which lead to $P(I) = 0$} separately.

We will prove the following theorem. 

\begin{theorem}\label{thm:ML-ub}
For any fixed biased metric and any $t \geq 0$, we have
$$\ell(D, t) \leq  \frac{P(I_t^c)}{2} \leq r(t).$$
In particular, this implies that ML has distortion at most $3$.
\end{theorem}  

By a theorem due to \cite{DBLP:conf/sigecom/GoelKM17}, no randomized or deterministic weighted tournament rule can have distortion better than 3, so ML is optimal among weighted tournament rules. 

\begin{proof}[Proof of \Cref{thm:ML-ub}]

For ease of notation, let us fix $t$ and {write} $I = I_t$. {We will frequently use the following fact: if $j \notin I$, then $$s_{I \cg j} \leq \min_{i \in I} s_{i \cg j} \leq s_{D(I) \cg j}.$$ Intuitively, if we consider the sets $S_{i \cg j}$ for $i \in I$, then from left to right, the terms above correspond to the proportion of voters in the intersection of the sets, the smallest set, and a randomly chosen set (in expectation) respectively.} \\

Let us first prove the theorem in the cases where $P(I)$ or $P(I^c)$ are zero, so that afterwards we can assume that $D(I)$ and $D(I^c)$ are well-defined.

If $P(I^c) = 0$, then we simply have that $\ell(D, t) = \frac{P(I^c)}{2} = 0$, and the theorem easily follows. If $P(I) = 0$, we need to show that $\ell(D, t) \leq \frac12 \leq r(t).$ Then {using the fact that $i^*\in I$, we have} that 
$$\ell(D, t) = \sum_{j \notin I} s_{I \cg j} p_j\leq \sum_{j \notin I} s_{i^* \cg j} p_j = s_{i^* \cg D}.$$
On the other hand, we have
$$r(t) \geq 1 - s_{i^* \cg I^c} \geq 1 - \min _{j \notin I}s_{i^* \cg j} \geq 1 - \sum_{j\in I^c}  s_{i^* \cg j} p_j = 1 - s_{i^* \cg D}.$$
We have that $s_{i^* \cg D} \leq \frac12$ since $D$ weakly beats 
the strategy of deterministically picking $i^*$, so $\ell(D, t) \leq \frac12 \leq r(t)$. {This proves the theorem for the cases in which one of $P(I)$ or $P(I^c)$ is zero}.\\

Henceforth, {let us} assume that $D(I)$ and $D(I^c)$ are well-defined. First, using the fact that $s_{I\cg j} \leq \displaystyle\min_{i\in I} s_{i \cg j} \leq s_{D(I) \cg j}$, we have
$$\ell(D, t) = \sum_{j \notin I} s_{I \cg j} p_j\leq \sum_{j \notin I} s_{D(I) \cg j} p_j = s_{D(I) \cg D(I^c)}\cdot P(I^c).$$
Similarly, using  $s_{i^*\cg I^c} \leq \displaystyle\min_{j\notin I} s_{i^*\cg j} \leq s_{i^*\cg D(I^c)}$, we have 
$$r(t) \geq 1 - s_{i^*\cg I^c} \geq 1 - s_{i^*\cg D(I^c)}.$$
Therefore, it suffices to show that 
\begin{equation}\label{eq:ML-suffices}
s_{D(I) \cg D(I^c)}\cdot P(I^c) \leq \frac{P(I^c)}{2} \leq 1 - s_{i^*\cg D(I^c)}.
\end{equation}

To prove this, we will rely on two somewhat general properties of any equilibrium. The first claim, below, is equivalent to the first inequality in \cref{eq:ML-suffices}.

\begin{claim} $s_{D(I) \cg D(I^c)} \leq \frac12$. \end{claim}

\begin{proof}
Intuitively, this is true because if it were the case that $s_{D(I) \cg D(I^c)} > \frac12$ then $D(I)$ could strictly beat $D$, %
contradicting the fact that it is an equilibrium. 

The proof of the claim relies on two facts. For any distribution $X$ over the candidates, we have
\begin{enumerate}
\item  $s_{X \cg D} \leq \frac12$, and 
\item $s_{X \cg X} = \frac12$.
\end{enumerate}
The first follows by definition of $D$ being an equilibrium, and the second follows by symmetry. Then by the law of conditional expectation,
\begin{align*}
\frac12 \geq s_{D(I) \cg D} &= P(I^c)s_{D(I)\cg D(I^c)} + P(I)s_{D(I)\cg D(I)}\\
&= P(I^c)s_{D(I)\cg D(I^c)} + (1 - P(I^c)) \cdot \frac12
\end{align*}
which means that $s_{D(I) \cg D(I^c)} \leq \frac12$ as claimed.
\end{proof}

Note that, applying the claim with $I$ replaced with $I^c$ we can in fact conclude that $s_{D(I) \cg D(I^c)} = s_{D(I^c) \cg D(I)} = \frac12$. 

Now it remains to show that 
$$\frac{P(I^c)}{2} \leq 1 - s_{i^* \cg D(I^c)}.$$
For brevity, let $p = P(I^c)$ and $q = s_{i^* \cg D(I^c)}$. We want to show that $\frac{p}{2} \leq 1 - q$. We will also introduce $r = s_{D(I) \cg i^*}$.
\begin{figure}[h!]
\centering
\includegraphics[scale=1]{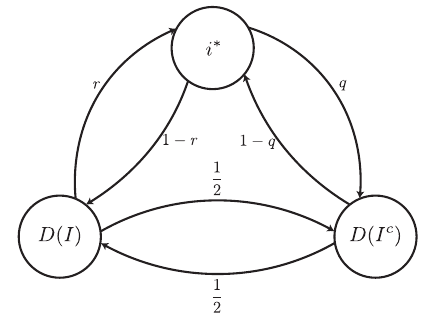}
\caption{Three strategies $D(I), D(I^c), i^*$. The edge $(A, B)$ is labeled with $s_{A\cg B}$.}
\end{figure}
Consider three different strategies for the game: $D(I)$, $D(I^c)$, and just deterministically choosing $i^*$. Then by conditional expectation and the assumption that $D$ is an equilibrium, 
$$\frac12 \leq s_{D \cg i^*} = P(I^c)s_{D(I^c)\cg i^*} + P(I)s_{ D(I) \cg i^*} = p(1 - q) + (1 - p)r.$$
On the other hand, the payoffs for these strategies satisfy a kind of triangle inequality, by the following claim.

\begin{claim}
For three strategies $A, B, C$, we have $s_{A\cg B} \leq s_{A\cg C} + s_{C\cg B}$.
\end{claim}

\begin{proof}
We have that $s_{A \cg B} = \displaystyle\Ev_{i\sim A, j\sim B, v\sim V}[\textbf{1}[i\cg_v j]]$, so $$- s_{A\cg B} + s_{A\cg C} + s_{C\cg B} = \Ev_{i\sim A, j\sim B, k\sim C, v\sim V}[-\textbf{1}[i\cg_v j] + \textbf{1}[i\cg_v k] + \textbf{1}[k\cg_v j]].$$
We claim that the term in the expectation is always nonnegative. If $i, j, k$ are all different then this is clear because it is not possible that $i\cg_v j$ but then $k \cg_v i$ and $j \cg_v k$. If  $j = k$ or $i = k$ then the first term cancels with the second or third terms. If $i = j$ then this term is
$$ -\textbf{1}[i\cg_v i] + \textbf{1}[i\cg_v k] + \textbf{1}[k\cg_v i]] = \tfrac12.$$
This means that the term in the expectation is always nonnegative, which proves the claim. \end{proof}

Now, applying the claim to our three strategies,
$$s_{D(I) \cg i^*} \leq s_{D(I) \cg D(I^c) } + s_{D(I^c) \cg i^*},$$
{which implies}
$$r \leq \tfrac12 + (1 - q).$$
Combining this with our previous inequality relating $p,q,r$ we have
$$\frac12 \leq p(1 - q) + (1 - p)r \leq p(1 - q) + (1 - p)(\tfrac32 - q) = (1 - q) + \frac{1 - p}{2}$$
which implies that $\frac{p}{2} \leq 1 - q$ as desired. {With this, we complete} the proof of \Cref{thm:ML-ub}.
\end{proof}

\section{Two Rules that Complement Maximal Lotteries}\label{sec:mechs}

In this section, we introduce two novel social choice rules, each of which can be mixed with Maximal Lotteries to get distortion better than 3: Random Consensus Builder (RCB) and Random Dictatorship on the Uncovered Set (RaDiUS). The two rules are very much cousins of one another, and their analyses have many commonalities. RCB has a slightly worse distortion guarantee, but we include it for several reasons: First, its analysis mirrors that of RaDiUS but is simpler, making it a natural warm-up for the tighter bound given by RaDiUS. Moreover, the final specification of the mixed rule using RCB is easier to state, and the distortion bound ends up being a clean algebraic number, $2\sqrt{2}$. These nice properties, along with the fact that RCB itself has a clean interpretation, make us believe that RCB, as a voting rule, can be of independent interest.

We will analyze the rules using our biased metric framework. Intuitively, the constraints will be harder to satisfy when $r(t)$ is often small, which makes $R$ small. However, when $r(t)$ is small, we can show that the election instance and the metric admit a certain structure which can be leveraged in the analysis. This structure is characterized by the following definition. {It} interplays with the function $r(t)$ via the subsequent proposition{, which} is actually all we need from $r(t)$ in this section. 

\begin{definition}
Given an election instance, a biased metric is $(\alpha, \beta)$-\emph{consistent} if whenever $s_{k \cg i^*} \geq \beta$, we have $x_k \leq \alpha R$.
\end{definition}

{The motivation behind the naming is as follows. If $i^*$ is the optimal candidate, and another candidate $k$ is preferred over $i^*$ by a large margin, then it is reasonable to expect that $k$ is a nearly optimal candidate as well. The $(\alpha, \beta)$-consistency of a metric measures how closely it adheres to this expectation, with $\alpha$ and $\beta$ specifying what ``nearly optimal'' and ``large margin'' entail.}

\begin{proposition}\label{prop:small-r-consistent-metric}
If $r(\alpha R) < \beta$ then the metric is $(\alpha, \beta)$-consistent.
\end{proposition}

\begin{proof}
$r(\alpha R) < \beta$ means $s_{\forall i\cg j,  x_i - x_j \leq \alpha R} > 1 - \beta$. Then, whenever $s_{k \cg i^*} \geq \beta$, we have $s_{k \cg i^*} + s_{\forall i\cg j,  x_i - x_j \leq \alpha R} > 1$, which means that there exists a voter $v$ such that $k \cg_v i^*$ and whenever $i \cg_v j$, we have $x_i - x_j \leq \alpha R$. {With $i=k$ and $j=i^*$, this implies that $x_k \leq x_{i^*} + \alpha R = \alpha R$}.
\end{proof}

Once the election instance and metric space are fixed, we will analyze both rules under the assumption that the metric is $(\alpha, \beta)$-consistent. This assumption will crucially come into play in \Cref{sec:mix} where we want these rules to perform particularly well when $\alpha$ is close to 0 and $\beta$ is close to $\frac12$. The results in this section will still extend to the general case by the following proposition.

\begin{proposition}\label{prop:all-biased-consistent}
All biased metrics are $(\frac1\beta, \beta)$-consistent for all $\beta \in (0, 1)$.
\end{proposition} 

\begin{proof}
{Suppose that} $s_{k \cg i^*} \geq \beta$. {If $v\in S_{k \cg i^*}$ then by \Cref{def:biased}, $2d(i^*, v) \geq x_k - x_{i^*} = x_k$. It follows that}
$$R =  2\SC(i^*) \geq {x_k}\cdot s_{k\cg i^*} \geq x_k \cdot\beta.$$
This means that $x_k \leq \frac1\beta R$, and the proposition follows.
\end{proof}

Both of our rules, RCB and RaDiUS, are parameterized by a tunable value $\beta \in (\frac12 , 1).$ Very roughly speaking, they both construct the graph on candidates {such that} $(i, j)$ is an edge whenever $s_{i \cg j} \geq \beta$, with the goal of choosing candidates that can always reach the low-distortion candidates in few hops. Note that the rules have natural interpretations when $\beta = \frac12$ and $\beta = 1$ which we will briefly discuss, but to avoid dealing with these (ultimately irrelevant) edge cases in the proofs, we will restrict $\beta$ to the open interval. %

\subsection{Random Consensus Builder}\label{ssec:rcb}

Below is the description of our first rule. {Let us first introduce the algorithm informally, along with some terminology (in italics) that will facilitate exposition. The algorithm starts by choosing a voter $v$ uniformly at random, whom we call the \emph{consensus builder}. The candidates are \emph{processed} in increasing order of $v$'s preference $\cg_v$. When a candidate $i$ is processed, the algorithm \emph{eliminates} each candidate $j$ that $v$ prefers over $i$, but a large fraction of voters has the opposite preference. Finally, the last processed candidate is the winner. Elimination effectively deletes candidates from the election---they are not processed and cannot win. When a candidate $j$ is eliminated in the same step that $i$ is processed, we say that $i$ \emph{eliminates} $j$.}

\begin{framed}
    \begin{center}
        \textbf{$\beta$-Random Consensus Builder (RCB)}
    \end{center}

    \begin{itemize}
        \item {Choose a voter $v$ uniformly at random}.

        \item Until {all candidates are either processed or eliminated}:

        \begin{itemize}

            \item {Let $i$ be $v$'s least favorite candidate that is neither processed nor eliminated}. 

        	\item Eliminate {each candidate} $j$ such that $j \cg_v i$ and $s_{i \cg j} \geq  \beta$. {Process $i$}.

       	\end{itemize}

        \item {Output the last candidate that is processed.}
         
    \end{itemize}

\end{framed}

One interpretation of this rule is that {the randomly chosen consensus builder $v$'s} preferences are the main guide of the rule {(and it is inclined to choose candidates that $v$ prefers)}, but if the voters as a whole have consensus disagreements with $v$'s preferences, then their view overrules. %
The threshold for how strong the consensus needs to be is tuned by {the} parameter $\beta$.

{Another perspective on this interpretation comes from} considering how {the rule operates} at the extremes {at which} $\beta = \frac12$ and $\beta = 1$. If $\beta = 1$, the chosen candidate is always $v$'s top choice, and so the rule is exactly Random Dictatorship. On the other hand, if $\beta = \frac12$, then it is not hard to see that if the chosen candidate is $a$, for every other candidate $b$, either $a$ defeats $b$ (i.e.{,} $s_{a \cg b} \geq \frac12$) or $a$ defeats a candidate who defeats $b$. The candidates that satisfy this property are called the \emph{uncovered set}. Some {well-known} voting rules always output a candidate in the uncovered set, including the Copeland rule. In this sense, $\beta$ {is also} a measure of interpolation {between rules based on the uncovered set and Random Dictatorship}.

We will show the following theorem.

\begin{theorem}\label{thm:rcb-ub}
{Consider} an election instance with an $(\alpha,\beta)$-consistent underlying metric. Then{,} if $D$ is the distribution output {of} $\beta$-RCB, we have
$$L(D) \leq (\alpha + \beta)R.$$
\end{theorem}

\begin{proof}[Proof of \Cref{thm:rcb-ub}]
Suppose $v$ is the consensus builder {randomly chosen at the start of the algorithm}. Let $j_v$ be the candidate that $\beta$-RCB picks. Note that each candidate is either {processed}, or is eliminated by some other {(processed)} candidate. {Also, since the algorithm processes the candidates in order of $v$'s preference with $j_v$ last, $j_v$ is $v$'s favorite among processed candidates.}

If $i^*$ is {processed} then let $k_v = i^*$, and otherwise let $k_v$ be the candidate that eliminates $i^*$. %
In order to prove the theorem, we will use the following three critical properties of $k_v$:

\begin{enumerate}[label=(\Roman*)]

\item $x_{k_v} \leq \alpha R$, \label{enum:1}

\item $j_v \cgeq_v k_v$, \label{enum:2}

\item $s_{k_v \cg j_v} < \beta$. \label{enum:3}

\end{enumerate}
\ref{enum:1} follows because either $k_v = i^*$ in which case $x_{k_v} = 0$, or $k_v$ eliminates $i$ which means that $s_{k_v \cg i^*}\geq \beta$ and so by the fact that the metric is $(\alpha, \beta)$-consistent, we have $x_{k_v} \leq \alpha R$. \ref{enum:2} follows because both $k_v$ and $j_v$ are {processed}, {and $j_v$ is $v$'s favorite processed candidate}. \ref{enum:3} follows because {when} $k_v$ is {processed}, either $k_v = j_v$ in which case $s_{k_v \cg j_v} = 0 < \beta$, or $k_v$ did not eliminate $j_v$ which means $s_{k_v \cg j_v} < \beta$.

We will use \ref{enum:1} and \ref{enum:3} to get a good upper bound on $L(D)$, and \ref{enum:1} and \ref{enum:2} to get a good lower bound on $R$. {In particular, we can show the following sequence of inequalities bounding the cost of $j_v$ compared to the cost of $i^*$.}
\begin{align*}
\SC(j_v) - \SC(i^*) &\leq s_{j_v \cgeq k_v} \min(x_{k_v}, x_{j_v}) + s_{k_v \cg j_v} x_{j_v}\\
&\leq (1 -\beta) \min(x_{k_v}, x_{j_v}) + \beta x_{j_v}\\
&\leq (1 -\beta) \alpha R + \beta x_{j_v}.
\end{align*}
{The first line comes from applying \Cref{def:biased}: for any voter $u$ we have $d(j_v, u) - d(i^*, u) \leq x_{j_v}$, and if $u \in S_{j_v \cgeq k_v}$ then $d(j_v, u) - d(i^*, u) \leq x_{k_v}$.}
The second line follows from the first because $x_{j_v} \geq \min(x_{k_v}, x_{j_v})$ and so the expression is maximized when $s_{k_v \cg j_v}$ is as large as possible {(which is $\beta$ by \ref{enum:3}). The final line uses \ref{enum:1}}. It follows that 
$$L(D) = \frac1n \sum_{v\in V}(\SC(j_v) - \SC(i^*)) \leq \alpha(1 - \beta) R + \beta \cdot \frac1n \sum_{v\in V}x_{j_v}.$$

On the other hand, since $j_v \cgeq_v k_v$, we have $2d(v, i^*) \geq x_{j_v} - x_{k_v} \geq x_{j_v} - \alpha R $. It follows that
$$R = 2\SC(i^*) \geq -\alpha R + \frac1n\sum_{v\in V}x_{j_v},$$ {which is equivalent to}  $$\frac1n\sum_{v\in V}x_{j_v} \leq (1 + \alpha)R.$$
Plugging this into our upper bound on $L(D)$, we get  
$$L(D) \leq  \alpha(1 - \beta) R + \beta(1 + \alpha)R = (\alpha + \beta) R$$
as desired.
\end{proof}

{The following is an immediate corollary of \Cref{thm:rcb-ub} after applying \Cref{prop:all-biased-consistent}.}

\begin{corollary}\label{cor:rcb-ub}
$\beta$-RCB guarantees distortion at most $1 + 2(\beta + \frac1\beta)$. 
\end{corollary}

\subsection{Random Dictatorship on the (Weighted) Uncovered Set} \label{ssec:radius}

Next, we consider a rule similar in spirit to RCB, but with better distortion guarantees.

\begin{framed}
    \begin{center}
        \textbf{$\beta$-Random Dictatorship on the (Weighted) Uncovered Set (RaDiUS)}
    \end{center}
    
    \begin{itemize}

    	\item Say that $a$ \emph{covers} $b$ if $s_{a\cg b} \geq  \beta$ and for any $c$ {such that $s_{c \succ a} \geq \beta$, $s_{c \succ b} \geq \beta$}.

        \item Let $U$ be the set of candidates that are not covered by any other candidate. 

        \item Choose a {voter uniformly at random} and output their favorite candidate in $U$.

     \end{itemize}

\end{framed}

The set $U$ was previously considered by \cite{DBLP:conf/ec/MunagalaW19} in the context of deterministic rules. They showed that there exists a $\beta$ such that any candidate from the set $U$ (which they called the \emph{weighted uncovered set}) has distortion at most $2 + \sqrt{5}$.

To see how this rule is similar to $\beta$-RCB, consider the following proposition. It also conveniently gives a proof that the set $U$ is always non-empty (which was proved in a different way in \cite{DBLP:conf/ec/MunagalaW19}).

\begin{proposition}\label{prop:rcb-from-U}
Suppose that $U$ is the weighted uncovered set constructed by $\beta$-RaDiUS. Then $\beta$-RCB always outputs a member from $U$.
\end{proposition}

\begin{proof}
{Suppose} towards a contradiction that for some voter $v$, the candidate $j_v$ chosen by RCB is covered by some other candidate $a$. {Candidate $a$ is either processed, or is eliminated by some (processed) candidate $c$.}

{If $a$ is processed, then since $j_v$ is $v$'s favorite processed candidate, we must have that $j_v \cg_v a$. But since $a$ covers $j_v$ we also have that $s_{a \cg j_v} \geq \beta$. But then $a$ would have eliminated $j_v$, which is a contradiction.}

{On the other hand if $a$ is eliminated by $c$, then since $c$ is processed we have that $j_v \cg_v c$. Since $c$ eliminates $a$ we have that $s_{c \cg a} \geq \beta$, but then the fact that $a$ covers $j_v$ implies that $s_{c \cg j_v} \geq \beta$. Together, these conditions imply that $c$ must eliminate $j_v$, which is also a contradiction.}
\end{proof}

Before we get into the distortion guarantee {of} $\beta$-RaDiUS, we prove two more facts that will be helpful.

\begin{proposition}
The covering relation is transitive.
\end{proposition}

\begin{proof}

Suppose $a$ covers $b$ and $b$ covers $c$. We claim that $a$ covers $c$. Since $b$ covers $c$ and $s_{a\cg b} \geq  \beta$ we have $s_{a \cg c} \geq  \beta$. Now suppose that for some $d$, $s_{d\cg a} \geq  \beta$. Then since $a$ covers $b$ we have $s_{d \cg b} \geq  \beta$ {and because} $b$ covers $c$ we have $s_{d\cg c} \geq \beta$. So indeed, $a$ covers $c$.
\end{proof}

\begin{proposition}\label{prop:covered-by-U}
If a candidate is not in $U$ then it is covered by a candidate in $U$.
\end{proposition}

\begin{proof}
Suppose that we build a graph on the candidates {such that} $(a, b)$ is an edge if $a$ covers $b$. We cannot have a cycle in this graph, because by transitivity this would imply that some candidate $i$ covers itself, which would imply the impossible $s_{i\cg i} \geq \beta > \frac12 > 0$.

If a candidate $i$ is not in $U$, it must have positive in-degree. Since the graph is acyclic, by arbitrarily following edges backwards from $i$, we must eventually reach a candidate $j$ with in-degree zero. {It follows} that $j \in U$ and there is a path from $j$ to $i$, which by transitivity means that $j$ covers $i$. 
\end{proof}

Now we prove the following distortion guarantee.

\begin{theorem}\label{thm:radius-ub}
{Consider} an election instance with an $(\alpha,\beta)$-consistent underlying metric. Then{,} if $D$ is the distribution output {of} $\beta$-RaDiUS, we have
$$L(D) \leq (\alpha(1 - \beta^2) + \beta)R.$$
\end{theorem}

The proof is similar in structure to the proof of \Cref{thm:rcb-ub}. {In that proof, $j_v$ was the winning candidate, and we defined a candidate $k_v$ satisfying the following three properties:
\begin{enumerate}[label=(\Roman*)]

\item  $x_{k_v} \leq \alpha R$, \label{enum':1}

\item $j_v \cgeq_v k_v$, \label{enum':2}

\item $s_{k_v \cg j_v} < \beta$. \label{enum':3}

\end{enumerate}}

The key difference is that rather than using the same candidate $k_v$ which satisfies the properties \ref{enum':1}, \ref{enum':2}, \ref{enum':3}, we will have one candidate $k_v$ which satisfies properties \ref{enum':1} and \ref{enum':3} and another candidate $k^*$ that satisfies properties \ref{enum':1} and \ref{enum':2}. The advantage is that in the {latter} case, we have the same candidate $k^*$ for \emph{all} voters $v$, which will allow us to get a stronger lower bound on $R$. However, having different candidates for the two cases makes the argument a little more complicated.

\begin{proof}[Proof of \Cref{thm:radius-ub}]
Once again, let $j_v$ be the candidate that is output when $v$ is the randomly chosen voter. {Let us} assume that $j_v \neq i^*${;} otherwise the rule picks the best candidate and all of the bounds will only be improved. Then since $j_v \in U${,} it must be the case that $i^*$ \emph{does not} cover $j_v$. Unpacking the definition, this means that either 
\begin{enumerate}[label=(\alph*)]
\item $s_{i^*\cg j_v} < \beta$, or 

\item there exists some $k$ such that $s_{k \cg i^*} \geq \beta$ but $s_{k \cg j_v} < \beta$.
\end{enumerate} 

Define $k_v$ so that $k_v = i^*$ if (a) occurs and $k_v = k$ if (b) occurs. In either case we have once again that $x_{k_v} \leq \alpha R$ and $s_{k_v \cg j_v} < \beta$. These are the properties \ref{enum':1} and \ref{enum':3}, and so by an identical argument we can show that
$$L(D) = \frac1n \sum_{v\in V}(\SC(j_v) - \SC(i^*)) \leq \alpha(1 - \beta) R + \beta \cdot \frac1n \sum_{v\in V}x_{j_v}.$$

Now define $k^*$ so that $k^* = i^*$ if $i^* \in U$ and otherwise, $k^*$ is some member of $U$ which covers $i^*$ (which exists by \Cref{prop:covered-by-U}). Once again, either $k^* = i^*$ and so $x_{k^*} = 0$, or $s_{k^*\cg i^*}\geq \beta$ and since the metric is $(\alpha, \beta)$-consistent we have $x_{k^*} \leq \alpha R$. In addition,  $j_v \cgeq_v k^*$, since $k^* \in U$ and $j_v$ is $v$'s favorite candidate in $U$. Thus, $k^*$ satisfies properties \ref{enum':1} and \ref{enum':2}.

It follows that for every voter $v$, 
$$2d(v, i^*) \geq  x_{j_v} - x_{k^*}.$$
Moreover, if $v$ satisfies $k^* \cgeq_v i^*$, the inequality can be stronger. In this case, $j_v \cgeq_v k^* \cgeq_v i^*$ and so
$$2d(v, i^*) \geq x_{j_v} = x_{k^*} + (x_{j_v} - x_{k^*}).$$
Since $k^* \cgeq_v i^*$ for at least a $\beta$ fraction of voters $v$, we have 
$$R = 2\SC(i^*) \geq \beta x_{k^*} + \frac{1}{n} \sum_{v\in V} (x_{j_v} - x_{k^*}) = -(1 - \beta)x_{k^*} +  \frac{1}{n} \sum_{v\in V} x_{j_v},$$
where we crucially use the fact that all voters share the same $k^*$.
It follows that 
$$\frac{1}{n} \sum_{v\in V} x_{j_v} \leq R + (1 - \beta)x_{k^*} \leq (1 + (1 - \beta)\alpha)R.$$
Plugging this into our upper bound on $L(D)$, we have
\begin{align*}
L(D) &\leq  \alpha(1 - \beta) R + \beta (1 + (1 - \beta)\alpha)R\\
&= (\alpha(1 - \beta^2) + \beta)R
\end{align*}
as claimed.
\end{proof}

{Like \Cref{cor:rcb-ub}, the following is an immediate corollary of \Cref{thm:radius-ub} after applying \Cref{prop:all-biased-consistent}.}

\begin{corollary}\label{cor:radius-ub}
$\beta$-RaDiUS guarantees distortion at most $1 + 2/\beta$. 
\end{corollary}

\section{Mixing Rules}\label{sec:mix}

Even though none of the three social choice rules we introduced can beat distortion 3 (see \Cref{sec:lbs}), it turns out that mixing them in a careful way can. 

Let us introduce the general technique for analyzing the mixture of these rules. Suppose that we have a given election instance and a biased metric. We would like to show that a particular rule achieves low distortion on this instance and metric. Consider the {plot} of $r(t)$ that is fixed by the metric. Let $\alpha(\cdot)$ be the function {over the domain $(\frac12, 1)$ that maps $\beta \mapsto {\frac{1}{R}}\min\{t: r(t) < \beta \}$}. Informally, if we draw the horizontal line $y = \beta$, then this line intersects the {curve} of $r(t)$ at the point $(\alpha(\beta)R, \beta)$. (If the intersection is a line segment, we take the rightmost point on the segment.)

Unsurprisingly, the function $\alpha(\cdot)$ is directly related to the $\alpha$ we were considering in \Cref{sec:mechs}: \Cref{prop:small-r-consistent-metric} tells us that if we have $\alpha(\cdot)$ corresponding to a biased metric, then the metric is $(\alpha(\beta), \beta)$-consistent for all $\beta\in(\frac12,1)$.

Moreover, we can use this function $\alpha(\beta)$ to get a tighter bound on the distortion of ML. Let $D_{\text{ML}}$ be the distribution output by ML. Then \Cref{thm:ML-ub} tells us that 
$$\ell(D_{\text{ML}}, t) \leq \frac{P(I_t^c)}{2} \leq r(t)$$
and since $P(I_t^c) \leq 1$, we have $\ell(D_{\text{ML}}, t) \leq \min(\frac12, r(t))$. On the other hand, the area that is below $r(t)$ but above the horizontal line $\frac12$ is exactly $R\int_\frac12^1 \alpha(\beta) \dif \beta$, and so it follows that
$$L(D_{\text{ML}}) + R\int_\frac12^1 \alpha(\beta) \dif \beta \leq R,$$
{or equivalently,}
\begin{equation}\label{eq:rcw-stronger} 
 L(D_{\text{ML}}) \leq \left(1 - \int_\frac12^1 \alpha(\beta) \dif \beta\right)R.
\end{equation}

\begin{figure}[!ht]
\centering
\includegraphics[scale=0.7]{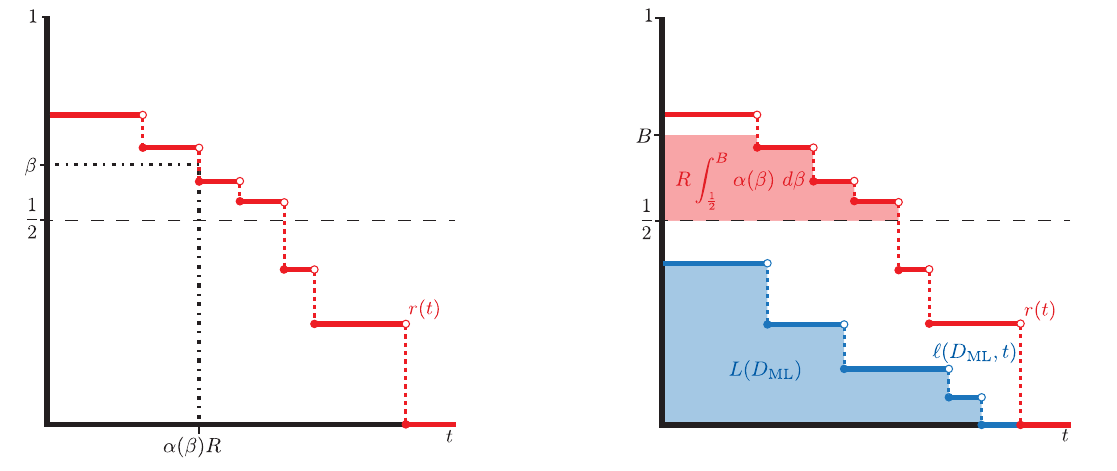}
\caption{Left: For each $\beta$, the horizontal line at $\beta$ intersects $r(t)$ at $(\alpha(\beta)R, \beta)$. Right: If the area above $\frac12$ and below $r(t)$ is large, we can get a better bound on $L(D_{\text{ML}})/R$. }
\end{figure}

This mixture of rules does well because, {in a sense}, ML and $\beta$-RCB/$\beta$-RaDiUS are complementary. For the analysis of ML not to %
have much wiggle room, the curve $r(t)$ should be above the line $\frac12$ very little. But then $\alpha(\beta)$ is small for values of $\beta$ that are slightly larger than $\frac12$, %
and with smaller $\alpha$ and $\beta$, we get much better guarantees in \Cref{thm:rcb-ub} and \Cref{thm:radius-ub}.

The rules will have three parameters: $p$, $B$, and $\rho(\cdot)$. Both rules run ML with probability $p$, and otherwise run $\beta$-RCB or $\beta$-RaDiUS where $\beta \in (\frac12, B)$ is drawn from a distribution with probability density function $\rho$. It will turn out that in the analysis, we will fix $p$ and $\rho$ as a function of $B$, and then to get the best distortion we just need to optimize a single variable function. In the warm-up, the right $\rho$ is just uniform, so we will omit that parameter.

{\Cref{thm:root-8} and \Cref{thm:2753} give the distortion guarantees that we can obtain by mixing ML with $\beta$-RCB and $\beta$-RaDiUS respectively. The mechanics of the two proofs are nearly identical, leveraging the distortion guarantees proved in \Cref{thm:rcb-ub} and \Cref{thm:radius-ub} for $\beta$-RCB and $\beta$-RaDiUS under the assumption that the metric is $(\alpha, \beta)$-consistent. The key difference is that \Cref{thm:radius-ub} gives a stronger bound, so this will directly lead to a stronger distortion guarantee in the end.}

\subsection{Warm-up: ML and RCB Get Distortion \texorpdfstring{$2\sqrt{2}$}{2√2}}

We note that the rule and proof might read more smoothly with the right values for $p$ and $B$ baked in, but we will keep these as variables to illustrate the mechanics of the proof technique.

\begin{framed}
    \begin{center}
        \textbf{ML mixed with RCB}
    \end{center}

    \begin{itemize}
        \item With probability $p$, run Maximal Lotteries.

        \item With probability $1 - p$, choose a \emph{uniformly} random $\beta \in (\frac12, B)$ and run $\beta$-Random Consensus Builder.
         
    \end{itemize}

\end{framed}

\begin{theorem}\label{thm:root-8}
With $p = \frac1{\sqrt{2}}$ and $B = \sqrt{2} - \frac12$, the rule ML mixed with RCB has distortion at most $2\sqrt{2} \approx 2.828$.
\end{theorem}

\begin{proof}
Suppose that we have a fixed election instance and metric space. Let $D_{\text{ML}}$ be the distribution output by ML, let $D_\beta$ be the distribution output by $\beta$-RCB, and let $D$ be the overall distribution of the rule. Note that $\beta$ is chosen according to the probability density function $\frac{1}{B - \frac12}$, so 
\begin{align*}
L(D) &= p L(D_{\text{ML}}) + (1 - p)\int_{\frac12}^{B}\frac{1}{B - \frac12}\cdot L(D_\beta) \dif \beta \\
&= p L(D_{\text{ML}}) + \frac{1 - p}{B - \frac12}\int_{\frac12}^{B} L(D_\beta) \dif \beta.
\end{align*}
Applying \cref{eq:rcw-stronger} {(changing the upper bound of the integral from 1 to $B$, which increases the right side)} and \Cref{thm:rcb-ub}, we get
\begin{align*}
\frac{L(D)}{R} &\leq p \left(1 - \int_\frac12^B \alpha(\beta) \dif \beta\right) + \frac{1 - p}{B - \frac12} \int_{\frac12}^B  (\alpha(\beta) + \beta) \dif \beta\\
&= p + \frac{1 - p}{B - \frac12}\int_{\frac12}^B \beta \dif \beta  + \left(-p + \frac{1 - p}{B - \frac12} \right)\int_{\frac12}^B \alpha(\beta) \dif \beta\\
&= p + \tfrac12(1 - p)(B  + \tfrac12)  + \left( \frac{1 - p(B + \frac12)}{B - \frac12} \right)\int_{\frac12}^B \alpha(\beta) \dif \beta.
\end{align*}
Now, $\alpha(\beta)$ is a function which depends on the metric, which could be chosen adversarially. However, we can completely eliminate this ``dangerous'' term by choosing $p$ and $B$ such that its coefficient is 0. This {step} is perhaps where the magic of the proof happens{---}by carefully balancing the two rules, we can get a kind of destructive interference that eliminates any dangerous terms. 

Choosing $p = \frac{1}{B + \frac12}$, we get
$$\frac{L(D)}{R} \leq \frac{1}{B + \frac12} + \tfrac12(B - \tfrac12).$$
It is not hard to check that choosing $B = \sqrt{2}-\frac12$ minimizes the above expression, at which point it is also $\sqrt{2} - \frac12$. This {choice} gives us distortion $1 + 2(\sqrt{2} - \frac12) = 2\sqrt{2}$.
\end{proof}

\subsection{ML and RaDiUS Get Distortion \texorpdfstring{$2.753$}{2.753}}

\begin{framed}
    \begin{center}
        \textbf{ML mixed with RaDiUS}
    \end{center}

    \begin{itemize}
        \item With probability $p$, run Maximal Lotteries.

        \item With probability $1 - p$, sample $\beta \in (\frac12, B)$ according to the probability density function $\rho(\beta)$ and run $\beta$-RaDiUS.
         
    \end{itemize}

\end{framed}

\begin{theorem}\label{thm:2753}
With appropriate choices for $p, B$, and $\rho(\cdot)$ the rule ML mixed with RaDiUS has distortion at most $2.753$.
\end{theorem}

\begin{proof}
Let $D_{\text{ML}}$ and $D_\beta$ be defined as in the proof of \Cref{thm:root-8}, but with $\beta$-RaDiUS in place of $\beta$-RCB. {Following the same line of reasoning, but now with \Cref{thm:radius-ub} in place of \Cref{thm:rcb-ub}},
\begin{align*}
\frac{L(D)}{R} &\leq p \left(1 - \int_\frac12^B \alpha(\beta) \dif \beta\right) + (1 - p)\int_{\frac12}^B \rho(\beta)\left(\alpha(\beta)(1 - \beta^2) + \beta\right) \dif \beta\\
&= p + (1 - p)\int_{\frac12}^B \rho(\beta)\beta \dif \beta  + \int_{\frac12}^B \alpha(\beta)\left(-p + \rho(\beta)(1 - p)(1 - \beta^2) \right) \dif \beta\\
&= 1- (1 - p)\int_{\frac12}^B \rho(\beta)(1 - \beta) \dif \beta + \int_{\frac12}^B \alpha(\beta)\left(-p + \rho(\beta)(1 - p)(1 - \beta^2) \right) \dif \beta.
\end{align*}
The last line uses the fact that $\displaystyle\int_{\frac12}^B \rho(\beta)\dif \beta = 1.$ In order to make the coefficient of $\alpha(\beta)$ equal to 0, we set 
$$\rho(\beta) = \frac{p}{(1 - p)(1 - \beta^2)}.$$
{We can use this to determine $p$ as a function of $B$. In particular, we have}
$$1 = \int_{\frac12}^B \rho(\beta) \dif \beta =\frac{p}{1- p}\int_{\frac12}^B \frac{\dif \beta}{1 - \beta^2},$$
{and so}
$$p = \frac{1}{1 + \displaystyle \int_{\frac12}^B \frac{\dif \beta}{1 - \beta^2}}.$$

With these choices, we have
\begin{align*}
\frac{L(D)}{R} &\leq 1 - p\int_{\frac12}^B \frac{\dif \beta}{1 + \beta}\\
&= 1 - \frac{\displaystyle\int_{\frac12}^B \frac{\dif \beta}{1 + \beta}}{1 + \displaystyle \int_{\frac12}^B \frac{\dif \beta}{1 - \beta^2}}\\
&= 1 - \frac{\ln\tfrac23 + \ln(1 + B)}{1 - \tfrac12\ln 3 + \tfrac12(\ln(1 + B) - \ln(1 - B))}.
\end{align*}
Using numerical optimization methods, we find that the best choice is $B \approx 0.876353$, which gives distortion $2.75271$. 
\end{proof}

\section{Discussion}\label{sec:discussion}

In this work, we studied {the Maximal Lotteries voting rule in the distortion framework,} and proposed {two novel simple rules,} Random Consensus Builder and RaDiUS. Using our biased metric framework, we show that a mix between ML and RCB has metric distortion at most $2 \sqrt{2}$, and a mix between ML and RaDiUS has distortion at most $2.753$. %

An immediate future direction is to further close the gap $(2.112, 2.753)$ of optimal metric distortion. Towards this, we propose the following ideas:
\begin{itemize}

\item Our RaDiUS rule is a hybrid of Random Dictatorship and a deterministic weighted tournament rule of \citetnostar{DBLP:conf/ec/MunagalaW19}. Is there a deterministic weighted tournament rule with {distortion better} than $2 + \sqrt{5} \approx 4.236$? Such a result is very interesting on its own, and can potentially serve as an ingredient for a rule with {distortion better} than $2.753$. (Note that our analysis for ML pinned down the optimal metric distortion for weighted tournament rules at $3$. For deterministic weighted tournament rules, the gap is $[3, 2 + \sqrt{5}]$.)
\item Our RaDiUS rule uses the notion of weighted uncovered set, which was designed to show a deterministic rule with good metric distortion. Would ideas that lead to distortion-optimal deterministic rules be useful, such as the matching uncovered set \citetnostar{DBLP:conf/ec/MunagalaW19} and related ideas \citetnostar{DBLP:conf/aaai/Kempe20a}, Plurality Matching \citetstar{DBLP:conf/focs/GkatzelisHS20}, Plurality Veto \citetnostar{DBLP:conf/ijcai/KizilkayaK22} and its variants \citetnostar{DBLP:conf/ijcai/KizilkayaK22,DBLP:conf/sigecom/Kizilkaya023}?
\item The biased metric framework potentially has more power than we have utilized in our proofs. \Cref{thm:rcb-ub,thm:radius-ub} show that for some function $f(\cdot,\cdot)$, their respective rules have distortion $1 + 2f(\alpha, \beta)$ under this assumption. If one can show a similar theorem for a new rule, but with a smaller function  $f(\cdot,\cdot)$, then this would improve the distortion upper bound. One could also attempt to go beyond our proof structure. For example, \cref{eq:weaker-biased} is a set of simpler and stricter constraints (derived by \citetnostar{DBLP:conf/soda/CharikarR22}) than what we use, but we do not know what distortion bounds we can achieve after this simplification. Further understanding the structures of biased metrics can be helpful in improving the metric distortion bounds.
\end{itemize}
Another intriguing direction is to find ``simpler'' rules that have good metric distortion:
\begin{itemize}
\item We managed to break the barrier of $3$ by mixing simple rules. Can we do this using an even simpler rule, e.g., one which does not look like a randomization between simple rules? A similar question can be asked for some non-metric distortion settings, where the Stable Lottery (or Stable Committee) rule, which looks like a randomization between two simple rules, gives optimal $\Theta(\sqrt{m})$ distortion \citepstar{DBLP:journals/teco/EbadianKPS24}.
\item Can we break the barrier of $3$ by using a minimal amount of randomness{---}for example, randomizing between at most two candidates, or only using randomness to sample a single voter (as RCB, RaDiUS, and Random Dictatorship do)?
\end{itemize}

\anonymize{
\section*{Acknowledgments}
The authors would like to thank Ashish Goel, Anupam Gupta, Kamesh Munagala, Aviad Rubinstein, and Michael Zhang for insightful discussions, and the anonymous reviewers for their helpful comments and suggestions.

Moses Charikar is supported by a Simons Investigator Award. Prasanna Ramakrishnan is supported by Moses Charikar's Simons Investigator Award and Li-Yang Tan's NSF awards 1942123, 2211237, and 2224246. Kangning Wang is supported by Aviad Rubinstein's NSF award CCF-2112824. Hongxun Wu is supported by ONR grant N00014-18-1-2562. This work was done while Kangning Wang was at Stanford University and the Simons Institute for the Theory of Computing, UC Berkeley.
}

\bibliography{ref}
\bibliographystyle{alpha}

\appendix

\section{Lower Bounds for RCB and RaDiUS}\label{sec:lbs}

In this section, we will show that $\beta$-RCB and $\beta$-RaDiUS always have worst-case distortion at least 3. In particular, we will show that \Cref{cor:radius-ub} is tight and \Cref{cor:rcb-ub} is almost tight. The reason for the gap between the upper and lower bounds for $\beta$-RCB is that \Cref{thm:rcb-ub} can be improved in some parameter regimes (though not in a way that improves the final distortion). Details on this are included in \Cref{apx:wuc-ub}.

\begin{theorem}\label{thm:radius-lb}
$\beta$-RaDiUS has distortion at least $1 + 2/\beta$.
\end{theorem}

\begin{proof}

Consider a family of instances in which the candidates are $i^*$, $k^*$, and a large set $U$. A $1 - \beta$ fraction of voters ranks $i^*\cg U \cg k^*$. When we write the set $U$ in this fashion, it means that these voters order the candidates of $U$ every way in equal proportion. The remaining $\beta$ fraction of voters has rankings of the form $j \cg k^* \cg i^* \cg U\setminus j$ for some $j \in U$. Each $j \in U$ is equally likely to be  $j$. 

\begin{figure}[!ht]
\centering
\includegraphics[scale=1]{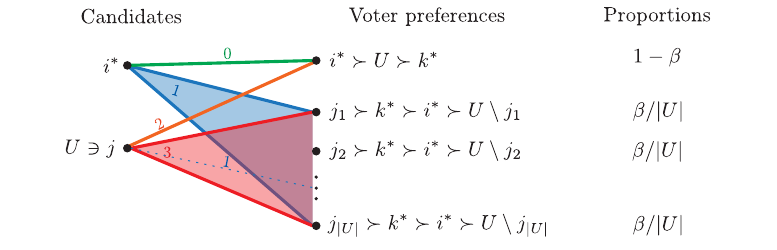}
\caption{The lower bound instance for $\beta$-RaDiUS.}
\end{figure}

Now we have that $s_{k^* \cg i^*} = \beta$, and for each $j \in U$, $s_{i^* \cg j} = 1 - \beta/|U|$ and $s_{k^* \cg j} = \beta(1 - \frac{1}{|U|})$. So in the graph in which $(a, b)$ is an edge if $s_{a \cg b}\geq \beta$, the only edges are $(k^*, i^*)$ and $(i^*, j)$ for each $j \in U$ (eventually we will take $|U|\to \infty$, so we will treat $|U|$ as large enough that $ 1 - \beta/|U| > \beta$). Then $k^*$ covers $i^*$, since there is an edge from $k^*$ to $i^*$ and $k^*$ has no in-edges. On the other hand, $i^*$ does not cover any $j \in U$, because there is no edge from $k^*$ to $j$. Thus, the weighted uncovered set is $U \cup \{k^*\}$. 

But then $\beta$-RaDiUS will always choose a candidate in $U$. Considering the biased metric for which $x_i = 2$ for all $i \neq i^*$ we have
\begin{align*}
\SC(i^*) &= \beta,\\
\SC(j) - \SC(i^*) &= 2\left(1 - \frac{\beta}{|U|}\right),
\end{align*}
for all $j \in U$. Taking $|U| \to \infty$, we get that 
$$\frac{\SC(j) - \SC(i^*)}{2\SC(i^*)} = \frac{1}{\beta}$$
which exactly corresponds to distortion $1 + 2/\beta$. 
\end{proof}

\begin{theorem}\label{thm:rcb-lb}
$\beta$-RCB has distortion at least $1 + 2/\beta$.
\end{theorem}

\begin{proof}
For $T$ sufficiently large, let $C_1, \ldots, C_{T+1}$ be disjoint sets of $T$ candidates each. In addition, we have one candidate $i^*$, and we will treat $\{i^*\}$ as $C_0$.

We will design the voters' preferences so that the election instance has the following properties:
\begin{enumerate}[label=(\arabic*)]
    \item $s_{c_{i + 1} \cg c_{i}} \geq \beta$ for each $c_i \in C_i$ and $c_{i + 1} \in C_{i + 1}$
    \item As $i$ increases, $s_{c_i \cg i^*}$ decreases until $s_{c_{T + 1} \cg i^*} \approx 0$.
\end{enumerate}
We will want that when any given voter is the chosen voter, some candidate in $C_1$ will eliminate $i^*$, and then some candidate in $C_2$ will eliminate everyone in $C_1$ above $i^*$ 
and so on, so that all the candidates above and including $i^*$ will get eliminated. In the end, we will get stuck with a candidate in $C_{T + 1}$ with the highest cost.

Consider the following voter preferences. A $1 - \beta - 1/T$ fraction of voters has a ranking of the form
$$ i^* \cg C_{T + 1} \cg C_{T} \cg \cdots \cg C_1.$$
(These sets should always be thought of as representing candidates ordered in every way equally often.) Another $1/T$ fraction of voters will have rankings of the form
$$ C_{T + 1} \cg C_{T} \cg \cdots \cg C_1 \cg i^*.$$
Then for each $1 \leq t \leq T$, a $\beta/T$ fraction of voters has a ranking of the form
$$(C_t \setminus c_t) \cg \cdots \cg (C_1 \setminus c_1)\cg i^* \cg C_{T + 1} \cg \cdots \cg C_{t + 1} \cg c_t \cg \cdots \cg c_1.$$
For each $i \leq t$, each member of $C_i$ is the candidate $c_i$ equally often over the preference lists of these voters. For brevity, we call the set of voters with rankings of the above form $S_{t}$. 

First, we will show that for $c_i \in C_i$ and $c_{i + 1} \in C_{i + 1}$, we have $s_{c_{i + 1} \cg c_i} \geq \beta$. First let us see this for $i = 0$. We have that $c_{1} \cg i^*$ for the $1/T$ fraction of voters in $S_{C_1 \cg i^*}$, and for all but a $1/T$ fraction of voters in $S_1 \cup \cdots \cup S_{T}$. It follows that 
$$s_{c_1 \cg i^*} = \frac1T + \beta\left(1 - \frac1T\right) = \beta + \frac{1 - \beta}{T} > \beta.$$
Now for $i > 0$, we have that $c_{i + 1} \cg c_i$ for the $1 - \beta$ fraction of voters outside of $S_1 \cup \cdots \cup S_{T}$. For the remaining voters, if $c_{i} \cg_v c_{i + 1}$, then one of two things must be true of $v$: Either $v \in S_{i}$, or $v \in S_{i'}$ for some $i' > i$ and the candidate $c_{i + 1}$ is the candidate from $C_{i + 1}$ who is below $i^*$. Each of these two events is true for at most a $1/T$ fraction of the voters in $S_1 \cup \cdots \cup S_{T}$, so it follows that 
$$s_{c_{i + 1} \cg c_i} \geq 1 - \beta + \beta(1 - \tfrac2T) = 1 - \tfrac2T \beta$$
and taking $T$ sufficiently large, this is at least $\beta$. 

For $1\leq i < j$ we have $s_{c_i \cg c_j} < \beta$ since $c_j \cg c_i$ for the $1 - \beta$ fraction of voters outside of  $S_1 \cup \cdots \cup S_{T}$, at at least one voter in $S_1$. These conclusions imply that when a voter outside of $S_{C_1 \cg i^*}$ is chosen by the rule, the candidates in $C_{T + 1}$ cannot be eliminated, and the candidates above and including $i^*$  will all get eliminated. It follows that with probability at least $1 - \frac1T$, the rule chooses a candidate in $C_{T + 1}$.

The biased metric we consider is the same as before, where $x_{i} = 2$ for $i \neq i^*$. With this we have:
\begin{align*}
\SC(i^*) &= \beta  + \frac1T,\\
\SC(j) - \SC(i^*) &= 2\left(1 - \frac1T\right),
\end{align*}
for $j \in C_{T + 1}$. Since the rule chooses a candidate in $C_{T+1}$ with probability at least $1 - \frac1T$, we get
$$\frac{\Ev[\SC(j) - \SC({i^*})]}{2\SC(i^*)} \geq \frac{(1 - \frac1T)^2}{\beta + \frac1T}$$
and taking $T \to \infty$, this corresponds to distortion $1 + 2/\beta$.
\end{proof}

\section{A Third Complementary Rule}\label{sec:indep-set}

In this section, we will give a brief discussion on another rule which is similar in flavor to $\beta$-RCB and $\beta$-RaDiUS. We chose not to include it in the main body of the paper because the guarantees are worse and the analysis is more complicated, but the rule itself is interesting and could be of independent interest. It may also give some insight into what led to $\beta$-RCB and $\beta$-RaDiUS. 

The rule is based on the following well-known combinatorial fact.

\begin{lemma}[\cite{chvatal1974every}]
For any directed graph $G = (V, E)$, there exists an independent set $U$ with the property that for all $v \in V$ there exists $u \in U$ such that there is a path of length at most $2$ from $u$ to $v$.
\end{lemma}

Such an independent set $U$ is called a \emph{quasi-kernel} (or sometimes \emph{semi-kernel} \cite{chvatal1974every}) of the directed graph $G$. %

\begin{proof}
The proof of \cite{chvatal1974every} is by induction, and it can be reworked into a nice algorithm to construct the set $U$. A similar observation has been made by \cite{croitoru2015note}. The algorithm goes as follows.

\begin{itemize}

	\item Arbitrarily order the vertices $v_1, v_2, \ldots, v_n$.

	\item For $i = 1, \ldots, n$:

	\begin{itemize}

		\item If $v_i$ is not eliminated, eliminate all $v_j$ such that $j > i$ and $(v_i, v_j)\in E$.

	\end{itemize}

	\item For $i = n, \ldots, 1$:

	\begin{itemize}

		\item If $v_i$ is not eliminated, eliminate all $v_j$ such that $j < i$ and $(v_i, v_j)\in E$.

	\end{itemize}

	\item The vertices that are not eliminated form the set $U$.

\end{itemize}

First, let us see why $U$ is an independent set. Suppose that there existed $v_i, v_j \in U$ such that there is an edge from $v_i$ to $v_j$. Since both are never eliminated, $v_i$ would have eliminated $v_j$ in the first loop if $i < j$ and in the second loop if $i > j$. This is a contradiction.

Next, suppose we have some $v \notin U$. Then either $v$ was eliminated by some $u \in U$, or by some $v' \notin U$ who was eliminated by some $u \in U$ ($v'$ eliminated $v$ in the first loop, then $u$ eliminated $v'$ in the second loop). In the first case, there is a path of length $1$ from $u$ to $v$, and in the second case there is a path of length 2. Thus, $U$ indeed satisfies the requirements.
\end{proof}

The key idea of the rule is to apply this lemma to the graph on candidates {that has} an edge from $a$ to $b$ {whenever} $s_{a \cg b} \geq \beta$, for some fixed $\beta \in (\frac12, 1)$. Then, even though we do not know $i^*$, we know there is some $j^* \in U$ such that $s_{j^* \cgeq k} \geq \beta$ and $s_{k \cgeq i^*} \geq \beta$ {for some $k$}.  On the other hand, for any $j \in U$, we have that $s_{j \cgeq j^*} \geq 1 - \beta$. {These conditions mimic those} we used in the proof of \Cref{thm:radius-ub} (there it is the same condition as if $j^* = k$) and so we can reason about it in a similar way. However, the extra distance between $j$ and $i^*$ (having $j^*$ and $k$ rather than just $k$) makes the analysis weaker.

Ultimately the rule is as follows.

\begin{framed}
    \begin{center}
        \textbf{$\beta$-Random Dictatorship on the {Quasi-Kernel} ($\beta$-{RDQK})}
    \end{center}
    
    \begin{itemize}

    	\item Fix an arbitrary ordering of the candidates $c_1, c_2, \ldots, c_m$.

		\item For $i = 1, \ldots, m$:

		\begin{itemize}

			\item If $c_i$ is not eliminated, eliminate all $c_j$ such that $j > i$ and $s_{c_i \cg c_j} \geq \beta$.

		\end{itemize}

		\item For $i = m, \ldots, 1$:

		\begin{itemize}

			\item If $c_i$ is not eliminated, eliminate all $c_j$ such that $j < i$ and $s_{c_i \cg c_j} \geq \beta$.

		\end{itemize}

		\item Let $U$ be the set of candidates that were never eliminated. Pick a uniformly random voter and choose their favorite candidate in $U$.

     \end{itemize}

\end{framed}

For the distortion guarantees, we need to make a slight modification to the definition of $(\alpha, \beta)$-{consistency}. Instead of the definition being whenever $s_{k \cg i^*} \geq \beta$ we have $x_k \leq \alpha R$, it will be whenever $s_{k \cg i} \geq \beta$ we have $x_k - x_i \leq \alpha R$. 

{Thanks to this change, we can use the fact that $s_{j^* \cgeq k} \geq \beta$ and $s_{k \cgeq i^*} \geq \beta$ to conclude} that $x_{j^*} \leq 2\alpha R$. Using the same approach as in the proofs of \Cref{thm:rcb-ub} and \Cref{thm:radius-ub}, we can establish the following theorem.

\begin{theorem}
Suppose that we have an election instance with an $(\alpha,\beta)$-consistent underlying metric. Then if $D$ is the distribution output by $\beta$-{RDQK}, we have
$$L(D) \leq (\alpha(1 - \beta)(2 + 3\beta) + \beta)R.$$
\end{theorem}

Note that this is always worse than \Cref{thm:radius-ub}, but when $\beta$ is close to 1, it is actually better than \Cref{thm:rcb-ub}. We will omit the proof of the theorem, since the mechanics are the same as the proofs we have seen before. 

Now, the fact that $\beta$-{RDQK} uses an arbitrary ordering of the candidates leaves open the possibility that by cleverly choosing the order of the candidates, one can get a better bound on the distortion. Indeed, all of the following modifications improve the guarantee, at least for some values of $\beta$:

\begin{enumerate}[label={(\arabic*)}]

	\item Choose a random voter, and use their preference list as the candidate order.

	\item Do the above, but then use the same random voter in the last step to select a candidate.

	\item Either of the above, but use the preference list in \emph{reverse} for the candidate order.

\end{enumerate}

In fact, (3) with the same voter in the last step is exactly $\beta$-RCB. (Since we consider candidates in reverse, the voter's favorite candidate in the set $U$ will be decided after the first loop.) As one might expect, figuring out the optimal way to randomize between all of these different options makes the rule and its analysis quite complicated. One can at least improve the upper bound on $L(D)/R$ to $\alpha(1 - \beta)(\tfrac12 + \tfrac92 \beta) + \beta$ by randomizing between (1) and (2), and randomizing between this, $\beta$-RCB, and ML one can get distortion a little less than $2 \sqrt{2}$. It turns out that $\beta$-RaDiUS gets a better guarantee than all of these for all values of $\alpha$ and $\beta$ so these do not lead to any improvement.

\section{Re-deriving Known Bounds via Biased Metrics}\label{sec:revisit-known} 

In this section, we revisit some of the rules that have been studied in the metric voting distortion literature, and show how upper bounds on their distortions can be proved using the biased metric framework introduced by \cite{DBLP:conf/soda/CharikarR22} and refined by our work. The fact that this can be done is unsurprising{---}since biased metrics {are} the hardest metrics, any distortion upper bound for a rule must somehow be arguing about them {under the hood}.  However, many of these upper bounds also essentially {redo} the work of defining the biased metrics (at least in some partial sense), and so having biased metrics as a primitive can simplify some of the proofs. We will demonstrate this for a handful of results, suggesting that biased {metrics} may be a useful primitive for future work.

\subsection{Distortion for (Weighted) Uncovered Set Rules}\label{apx:wuc-ub}

In \Cref{ssec:radius}, we showed an upper bound on the distortion of running Random Dictatorship restricted to the candidates in the weighted uncovered set $U$. In this section, we will show the following distortion bound on choosing \emph{any} candidate in $U$. 

\begin{theorem}\label{thm:wuc-ub}
Let $U$ be the $\beta$-weighted uncovered set for some $\beta \in [\frac12, 1)$. For any $(\alpha, \beta)$-consistent metric, any candidate $j\in U$ will satisfy 
$$\frac{\SC(j) - \SC(i^*)}{2\SC(i^*)} \leq \alpha\min\left(0,1 - \frac{\beta^2}{1 - \beta}\right) + \frac{\beta}{1 - \beta}.$$
\end{theorem}
Note that the coefficient of $\alpha$ is nonzero as long as $\beta \leq \varphi^{-1}$ where $\varphi = \frac{1 + \sqrt{5}}{2}$ is the golden ratio. Using $\alpha = \frac1\beta$, we get \Cref{cor:wuc-ub}. {Note that \Cref{cor:wuc-ub} is also a direct corollary of the more general \cite[Corollary~5.1]{DBLP:conf/aaai/Kempe20a}.}

\begin{corollary}\label{cor:wuc-ub}
For all metrics, a candidate in $U$ has distortion at most $1 + 2/\beta$ for $\beta\in[\frac12,  \varphi^{-1}]$, and at most $1 + 2\frac{\beta}{1 - \beta}$ for $\beta \in [\varphi^{-1},1)$.
\end{corollary}

Taking $\beta = \frac12$, this recovers the result of \cite{DBLP:conf/aaai/AnshelevichBP15,DBLP:journals/ai/AnshelevichBEPS18} that any rule that outputs a candidate in the (unweighted) uncovered set {(including the Copeland rule)} has distortion at most 5. Taking $\beta = \varphi^{-1}$ this also recovers the distortion $2 + \sqrt{5}$ rule due to \cite{DBLP:conf/ec/MunagalaW19}. 

We also note that by \Cref{prop:rcb-from-U}, these bounds will also apply to $\beta$-RCB. In particular, \Cref{cor:wuc-ub} improves \Cref{cor:rcb-ub} for $\beta\in [\frac12, 0.682]$, which explains why there is a gap between \Cref{cor:rcb-ub} and \Cref{thm:rcb-lb}. In fact, \Cref{thm:rcb-lb} is tight for $\beta\in[\frac12,  \varphi^{-1}]$. 

These bounds of course also apply to $\beta$-RaDiUS, but it turns out that {they do} not improve on \Cref{thm:radius-ub}. In particular, for $\alpha \leq \frac1\beta$, $\alpha(1 - \beta^2) + \beta$ is smaller than the expression in \Cref{thm:wuc-ub}.

\begin{proof}[Proof of \Cref{thm:wuc-ub}]
Let $j$ be any candidate in the $\beta$-weighted uncovered set. The fact that $j$ is not covered by $i^*$ implies that either $s_{j\cg i^*} \geq 1 - \beta$, or there exists $k$ such that $s_{k \cg i^*} \geq \beta$ and $s_{j \cg k} \geq 1 - \beta$. If the former occurs, we just let $k = i^*$ so that we have the single condition that $s_{k \cgeq i^*} \geq \beta$ and $s_{j \cg k} \geq 1 - \beta$. %

By a similar argument as in the proofs of \Cref{thm:rcb-ub} and \Cref{thm:radius-ub}, we can show that 
$$\SC(j) - \SC(i^*) \leq x_k + \beta\max(x_j - x_k, 0)$$
and 
$$2\SC(i^*) \geq \beta x_k +  (1 - \beta)\max(x_j - x_k, 0)$$
{which rearranges to}
$$\max(x_j - x_k, 0) \leq \frac1{1 -\beta} (2\SC(i^*) - \beta x_k).$$
It follows that 
$$\SC(j) - \SC(i^*) \leq x_k\left(1 - \frac{\beta^2}{1 - \beta}\right) + \frac{\beta}{1 - \beta} \cdot 2\SC(i^*).$$
If $1 - \frac{\beta^2}{1 - \beta} \geq 0$, then we can use $x_k \leq \alpha \cdot 2\SC(i^*)$. Otherwise, we just use $x_k \geq 0$. Putting these two cases together gives us the desired result. 
\end{proof}

\subsection{Matchings in Domination Graphs}

In the pursuit of a deterministic distortion 3 rule, \cite{DBLP:conf/ec/MunagalaW19} gave the following definition, and proved the subsequent theorem.

\begin{definition}
Given an election instance and a pair of candidates $a, b$, the bipartite graph $G(a, b)$ is defined as follows. Each side of the graph is a copy of the set of voters $V$. There is an edge from $v$ to $v'$ if there exists a candidate $c$ such that $a \cgeq_v c$ and $c \cgeq_{v'} b$.

The \emph{matching uncovered set} is the set of candidates $a$ such that for all $b \neq a$, $G(a, b)$ has a perfect matching.
\end{definition}

\begin{theorem}[\cite{DBLP:conf/ec/MunagalaW19}]\label{thm:mus-graph}
Every candidate in the matching uncovered set has distortion at most $3$.
\end{theorem}

Given this theorem, all that remains is to show that the matching uncovered set is always nonempty. \cite{DBLP:conf/focs/GkatzelisHS20} proved this by considering a more manageable definition, below.

\begin{definition}
Given an election instance and a candidate $a$, the \emph{domination graph} $G(a)$ is a bipartite graph which has a vertex on each side for each voter, and the edge $(v, v')$ exists if $a \cgeq_v \topC(v')$ where $\topC(v')$ is the favorite candidate of $v'$.
\end{definition}

It is not hard to see that $G(a)$ is a subgraph of $G(a, b)$ for all $b\neq a$, \Cref{thm:mus-graph} implies that if $G(a)$ has a perfect matching, then $a$ has distortion at most $3$. Since there is only one graph per candidate, this version can be easier to work with. \cite{DBLP:conf/focs/GkatzelisHS20} proved the upper bound of $3$ by showing that there exists $a$ such that $G(a)$ always has a perfect matching, which also implies that the matching uncovered set is always nonempty.

Here, we will give an alternate proof of \Cref{thm:mus-graph} via the biased metric framework.

\begin{proof}[Proof of \Cref{thm:mus-graph}]
Consider the constraints given by \cref{eq:fancy} with $\lambda = 1$ and setting $p_{{a}} = 1$ and $p_i = 0$ for $i \neq {a}$. {It follows that} $a$ achieves distortion 3 if for all subsets of voters $I$ with $a \notin I$, and all $i^* \in I$, we have $s_{I \cg a}  + s_{i^*\cg I^c} \leq 1$.

Suppose that $a$ is in the matching uncovered set. We will use the fact that $G(a, i^*)$ has a perfect matching to prove that $s_{I \cg a}  + s_{i^*\cg I^c}  \leq 1$. 

We claim that in $G(a, i^*)$, there is no edge $(v, v')$ such that $v \in S_{I \cg a}$ and $v' \in S_{i^*\cg I^c}$. If there were, then we have that there exists a candidate $c$ such that $a \cgeq_v c$ and $c \cgeq_{v'} i^*$. But then we have that $I \cg_v a\cgeq_v c$ which means that $c \notin I$, but also $c \cgeq_{v'} i^* \cg_{v'} I^c$ which means $c\in I$. This is a contradiction, so the claim is true. 

It follows that in $G(a, i^*)$, the neighbors of the set $S_{I \cg a}$ (on the left) are disjoint from the set $S_{i^* \cg I^c}$ (on the right), which means that $|N(S_{I\cg a})| + |S_{i^*\cg I^c}| \leq n$. Since $G(a, i^*)$ has a perfect matching, by Hall's theorem,  $|N(S_{I\cg a})| \geq |S_{I \cg a}|$. It follows that $s_{I \cg a}  + s_{i^*\cg I^c} \leq 1$ as desired.
\end{proof}

\subsection{Plurality Veto}

\cite{DBLP:conf/ijcai/KizilkayaK22} later introduced a novel voting rule, Plurality Veto, and showed that it has distortion at most 3 via \Cref{thm:mus-graph} (the domination graph version). Their rule, and the proof of its distortion{,} are very clean and simple. Here, we take their proof and translate it into one that goes through biased metrics instead of \Cref{thm:mus-graph}. This {version of the proof} shows that if one takes biased metrics as a primitive, one can prove that there exists a deterministic rule with distortion 3 within a couple of paragraphs.  

\begin{framed}
    \begin{center}
        \textbf{Plurality Veto (\cite{DBLP:conf/ijcai/KizilkayaK22})}
    \end{center}
    
    \begin{itemize}

    	\item Initially the score of candidate $i$ is $n\cdot\plu(i)$. 

    	\item One by one, each voter decrements the score of their least favorite candidate with positive score.

    	\item The last candidate with positive score wins.

     \end{itemize}

\end{framed}

\begin{theorem}[\cite{DBLP:conf/ijcai/KizilkayaK22}]
Plurality Veto guarantees distortion $3$.
\end{theorem}

\begin{proof}
Suppose we have an election instance, and let $c$ be the candidate {that} Plurality Veto chooses. Like before it suffices to show that for all sets $I$ such that $c\notin I$, and all $i^*\in I$, we have $s_{I \cg c} \leq 1 - s_{i^*\cg I^c}$.

The key observation is that the voters in $S_{I \cg c}$ do not decrement the score of any candidate in $I$. This is because $c$ always has positive score, and so for any voter in this set, no candidate in $I$ can be their least favorite candidate with positive score. On the other hand, since none of the candidates in $I$ are eventually chosen, at least $n\sum_{i\in I} \plu(i)$ voters must decrement the score of some candidate in $I$. It follows that 
$$s_{I\cg c} \leq 1 - \sum_{i\in I} \plu(i) \leq 1 - s_{i^* \cg I^c}$$
where the last inequality follows because the top candidate of a voter in $S_{i^* \cg I^c}$ must be in $I$. 
\end{proof}

\cite{DBLP:conf/ijcai/KizilkayaK22} also considered a class of randomized rules that are variants of Plurality Veto, called $k$-Round Plurality Veto. Instead of every voter having the opportunity to decrement some candidate's score, only $k$ voters do so. Then, the rule randomly chooses a candidate with probability proportional to their score. Notice that if $k = 0$, this rule is exactly Random Dictatorship, so the choice of $k$ can be thought of as a measure of interpolation between Random Dictatorship and Plurality Veto. 

Using a generalization of the flow technique used in \cite{DBLP:conf/aaai/Kempe20b}, \cite{DBLP:conf/ijcai/KizilkayaK22} showed that for any $k$, $k$-Round Plurality Veto has distortion at most 3. Below, we show that this can also be done via biased metrics, which is arguably simpler.

\begin{theorem}[\cite{DBLP:conf/ijcai/KizilkayaK22}]
$k$-Round Plurality Veto has distortion $3$ for all $0\leq k \leq n$.
\end{theorem}

\begin{proof}
Let $p_j$ be the probability that $k$-Round Plurality Veto chooses candidate $j$ (for some fixed ordering of the voters). {Once again using} \cref{eq:fancy}, it suffices to show that for all sets $I$ and all $i^* \in I$, we have that
$$\sum_{j \notin I} s_{I\cg j} p_j \leq 1 - \sum_{i\in I}\plu(i) = \sum_{j\notin I}\plu(j).$$
Among the $k$ voters that decremented a candidate's score, let $V_j$ denote the set that decremented candidate $j$'s score, and let $V_I = \bigcup_{i\in I} V_i$ denote the set that decremented the score of some candidate in set $I$. We will use the lower case $v$ to denote the proportion of voters in the corresponding set. For instance, note that $\frac{k}{n} = v_C =v_I + v_{I^c}$. With this notation, we can observe that
$$p_j = \frac{\plu(j) - v_j}{1 - v_C} =\frac{\plu(j) - v_j}{1 - v_I - v_{I^c}}.$$ 
{Observe that} if $p_j > 0$ then no voter in the set $S_{I \cg j}$ decrements the score of a candidate in $I$. i.e., $s_{I \cg j} \leq 1 - v_I$. This {observation} tells us that 
$$\sum_{j \notin I} s_{I\cg j} p_j  \leq \frac{1 - v_I}{1 - v_I - v_{I^c}}\sum_{j \notin I} (\plu(j) - v_{j})= \frac{1 - v_I}{1 - v_I - v_{I^c}} \left(-v_{I^c} + \sum_{j \notin I} \plu(j)\right).$$
Therefore, it suffices to show that
$$(1 - v_I)\left(-v_{I^c} + \sum_{j \notin I} \plu(j)\right) \leq (1 - v_I - v_{I^c})\sum_{j\notin I}\plu(j).$$
Rearranging, we see that this is equivalent to
$$\sum_{j\notin I}\plu(j) \leq 1 - v_I \iff v_I \leq  \sum_{i\in I}\plu(i).$$
This is of course true because $n\sum_{i\in I}\plu(i)$ is the total initial score of the candidates in $I$, and $nv_I = |V_I|$ is the amount {by which} these scores are decremented in the rule.
\end{proof}

\subsection{Random Dictatorship and its Variants}

\cite{DBLP:journals/jair/AnshelevichP17,DBLP:conf/sigecom/FeldmanFG16} first showed that Random Dictatorship, which chooses candidate $i$ with probability $\plu(i)$, gets distortion 3. This {result} can be proved quite easily using \cref{eq:fancy}:
$$\sum_{j \notin I} s_{I \cg j} \plu(j) \leq \sum_{j \notin I} \plu(j) \leq 1 - s_{i^* \cg I^c}.$$
The last inequality follows because {the set of voters who prefer $i^*$ over $I^c$ is disjoint from the set of voters who rank $j \notin I$ first, for any $j$}.

\cite{DBLP:conf/focs/GkatzelisHS20} showed that Smart Dictatorship, which chooses candidate $i$ with probability proportional to $\frac{\plu(i)}{1 - \plu(i)}$, has distortion at most $3 - 2/m$ within the class of instances with $m$ candidates. \cite{DBLP:conf/soda/CharikarR22} showed that this can also be proved using \cref{eq:fancy}. Since $1 - \plu(j) \geq s_{I \cg j}$, we have very similarly that
$$\sum_{j \notin I} s_{I \cg j} \frac{\plu(j)}{1 - \plu(j)} \leq \sum_{j \notin I} \plu(j) \leq 1 - s_{i^* \cg I^c}.$$
{It follows} that in \cref{eq:fancy}, we can take $\frac{1}{\lambda} = \sum_{i} \frac{\plu(i)}{1 - \plu(i)}$, which is at least $\frac{1}{1 - \|\textbf{plu}\|_2^2}$ by Jensen's inequality with the function $f(x) =\frac1{1 - x}$. {The resulting distortion is at most} $3 - 2\|\textbf{plu}\|_2^2 \leq 3 - 2/m$.

Finally, we note that this approach can also give us stronger distortion guarantees in special cases of election instances. For instance, we can prove the following.

\begin{theorem}\label{thm:rd-bwc}
Within the class of instances {for which} $s_{i\cg j} \in [\frac12 -\eps, \frac12 + \eps]$ for all $i, j$, Random Dictatorship guarantees distortion at most $2 + 2\eps$.
\end{theorem}

\begin{proof}
We have $s_{I \cg j} \leq \frac12 + \eps$, and so
$$\sum_{j \notin I} s_{I \cg j} \plu(j) \leq (\tfrac12 + \eps)\sum_{j \notin I} \plu(j) \leq (\tfrac12 + \eps)(1 - s_{i^* \cg I^c}).$$
{It follows} that in \cref{eq:fancy}, we can take $\lambda = \frac12 + \eps$, which corresponds to distortion $1 + 2(\frac12 + \eps) = 2 + 2\eps$. 
\end{proof}

\section{Proofs for Biased Metrics}\label{sec:biased-pfs}

In this section, we will prove two key properties of biased metrics: (1) that \Cref{def:biased} actually defines a valid metric space, and (2) that these metrics are the hardest for the problem. These proofs are adapted from \cite{DBLP:conf/soda/CharikarR22} with minor adjustments.

\begin{proposition}
For any vector $(x_1, \ldots, x_m)$ of nonnegative real numbers {such that} $x_{i^*} = 0$, and any election instance, the biased metric corresponding to $(x_1, \ldots, x_m)$ is indeed a valid distance metric. 
\end{proposition}

\begin{proof}
Clearly, the distances we have defined are nonnegative, since the expression for $d(i^*, v)$ allows for $i = j$, which means that it is the maximum over a set that includes 0, and for any $j \neq i^*$ we have $d(j, v) \geq d(i^*, v)$. Thus, it suffices to show that the metric satisfies the triangle inequality. 

We can view the metric as a weighted graph, and in order to show that it satisfies the triangle inequality, we need to show that the weight of any edge is at most the sum of the weights of any path between the endpoints of the edge.

Suppose we have some path between a candidate $j$ and voter $v$. We will show that $d(j, v)$ is at most the total weight of the path. Recall that $d(j, v) = d(i^*, v)  + \displaystyle\min_{k: j \cgeq_v k} x_k \leq d(i^*, v) + x_j$. Now, suppose that the first two edges on the path are $(j, u)$ and $(u, k)$, and the last edge is $(i, v)$. Then the total length of the path is at least 
$$d(j, u) + d(u, k) + d(i, v) \geq d(j, u) + d(i^*, u) + d(i^*, v).$$ 
Now, $d(j, u) = d(i^*, u) + x_\ell$ for some $\ell$ for which $j \cgeq_u \ell$. By the definition of $d(i^*, u)$, we also have $2d(i^*, u) \geq x_j - x_\ell$. Putting all of these together, we have
$$d(j, u) + d(i^*, u) + d(i^*, v) \geq 2d(i^*, u) + x_\ell + d(i^*, v) \geq x_j + d(i^*, v) \geq d(j, v)$$
as desired.
\end{proof}

\begin{proposition}\label{prop:suffbiased}
Suppose we have an election instance and a distance metric $d$ that is consistent with the instance. Let $i^* = \arg\min_i\SC(i, d)$. Then there is a biased metric $\hat{d}$ such that $\SC(i^*, \hat{d}) \leq \SC(i^*, d)$ and  $\SC(j, \hat{d}) - \SC(i^*, \hat{d}) \geq \SC(j, d) - \SC(i^*, d)$ for each $j \neq i^*$.
\end{proposition}

In particular, this shows that for any election rule, the distortion of the rule is {weakly larger} with $\hat{d}$ than with $d$. Thus, showing that a rule has low distortion on all biased metrics is sufficient to show that it has low distortion for all metrics.

\begin{proof}
Let $x_i = d(i, i^*)$, and let $\hat{d}$ be the biased metric for $(x_1,x_2, \ldots, x_m)$. We will show that for any voter $v$, $\hat{d}(i^*, v) \leq d(i^*, v)$ and  $\hat{d}(j, v) - \hat{d}(i^*, v) \geq d(j, v) - d(i^*, v)$, which will immediately imply the proposition.

Fix $v$, and let $i$ and $j$ be such that $i \cgeq_v j$. Then we have
$$ d(i, i^*) \leq d(i, v) + d(i^*, v) \leq d(j, v) + d(i^*, v) \leq d(j, i^*) + 2d(i^*, v).$$
{It follows} that $d(i^*, v) \geq  \frac{d(i, i^*) - d(j, i^*)}{2} = \frac{x_i - x_j}{2}$. Taking the maximum over all choices of $i$ and $j$ such that $i \cgeq_v j$, we get 
$$d(i^*, v) \geq \frac{1}{2}\max_{i, j: i \cgeq_v j} (x_i - x_j) = \hat{d}(i^*, v)$$ 
as claimed. Next, fix a candidate $j \neq i^*$ and a voter $v$. Let $k = \displaystyle\arg \min_{k: j \cgeq_v k} x_k$, so that $\hat{d}(j, v) - \hat{d}(i^*, v) = x_k$. Then we have
$$d(j, v) \leq d(k, v) \leq d(k, i^*) + d(i^*, v) = x_k + d(i^*, v)$$
which means that 
$$ d({j}, v) - d(i^*, v) \leq x_j = \hat{d}({j}, v) - \hat{d}(i^*, v)$$
as claimed. {With this, we establish} the proposition.
\end{proof}

\end{document}